\documentclass[aps,pra,superscriptaddress,nofootinbib,twocolumn]{revtex4-2}

\usepackage{amsmath,amsthm,amsfonts,amssymb,amscd} 
\usepackage[utf8]{inputenc}
\usepackage{indentfirst}
\usepackage{graphicx}
\usepackage{color}
\usepackage[T1]{fontenc}
\usepackage[american,ngerman,greek,brazilian,british]{babel}
\usepackage{float}
\usepackage[dvipsnames]{xcolor}
\usepackage[colorlinks=true,citecolor=Mulberry,linkcolor=prettygreen,urlcolor=MidnightBlue,hyperindex,breaklinks]{hyperref}
\usepackage{breakurl}
\usepackage{cleveref}
\usepackage{mathtools,xfrac}
\usepackage{bm}
\usepackage{ulem}
\usepackage{enumerate}
\usepackage{sidecap}
\usepackage{microtype}
\usepackage{enumitem}
\usepackage{bbold}
\usepackage{listings}
\usepackage{algorithm}
\usepackage{algpseudocode}
\usepackage{multirow}
\usepackage{setspace}
\usepackage{physics}
\usepackage{dsfont}
\usepackage{soul}
\usepackage{lipsum}  
\usepackage{comment}
\usepackage{graphicx}

\newcommand{\id}{\mathbb{1}}

\renewcommand{\norm}[1]{\left\lVert#1\right\rVert}

\newcommand{\EE}{\mathbb{E}}

\newcommand{\Hcal}{\mathcal{H}}

\newcommand{\Lcal}{\mathcal{L}}

\newcommand{\Qcal}{\mathcal{Q}}

\definecolor{awesome}{rgb}{0.93, 0.53, 0.18}
\definecolor{prettygreen}{RGB}{5,125,143}

\newtheorem{theorem}{Theorem}
\newtheorem*{theorem*}{Theorem}
\newtheorem{proposition}{Proposition}
\newtheorem{lemma}{Lemma}
\newtheorem{corollary}{Corollary}

\newcommand{\lipsix}{Sorbonne Universit\'e, CNRS, LIP6, F-75005 Paris, France}

\begin{document}

\title{Probabilistic and approximate universal quantum purification machines}

\author{Zoe Garc\'{i}a del Toro}
\email{zoe.garcia@lip6.fr}
\address{\lipsix}

\author{Jessica Bavaresco}
\email{jessica.bavaresco@lip6.fr}
\address{\lipsix}

\begin{abstract} 
We study the task of lifting arbitrary quantum states and channels to purifications and Stinespring dilations, respectively, in both the probabilistic exact and deterministic approximate settings. We formalize this task through a general framework of quantum purification machines that, given a finite number of copies or uses of a black-box input, aim to output a corresponding purification or Stinespring dilation. In the probabilistic exact setting, we show that universality is not necessary to rule out such transformations: the simple requirement that a machine purifies two inputs of different rank with non-zero probability already implies that it cannot be described by a linear positive map. This simple argument captures a fundamental obstruction of quantum theory and recovers the impossibility of universal probabilistic purification from finitely many copies. In the approximate setting, we allow for general machines that are not required, in general, to produce a pure output. Using the minimum average error as our figure of merit, we derive analytical expressions for the performance of several physically motivated strategies as well as a general upper bound on the achievable error, which is tight in a specific regime. Our analysis reveals a trade-off: strategies that produce a pure output---among which we prove the optimal to be a strategy that produces as a fixed output a maximally entangled purification of the fully depolarizing channel---perform optimally between those considered for large environment dimension, while append-environment strategies that generally produce non-pure outputs perform better at small environment dimension. 
\end{abstract}

\maketitle

\section{Introduction}\label{sec:: Introduction}

\begin{figure}[t]
\begin{center}
    \includegraphics[width=0.75\columnwidth]{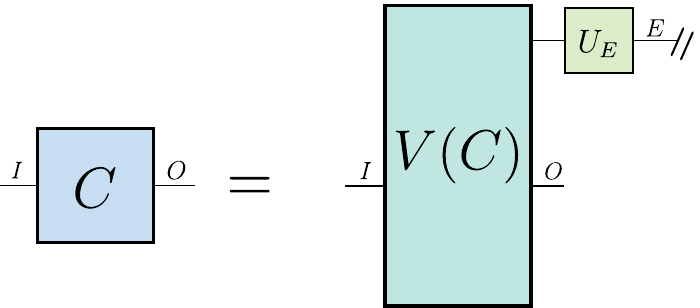}
    \caption{\textbf{Representation of the Stinespring dilation theorem.} Every quantum channel $C$ can be seen as an isometric channel $V(C)$, acting jointly on the original systems and an environment under a unitary freedom $U_E$, followed by a partial trace that discards the environment system. The purification of states is equivalent to the purification of quantum channels that have a trivial input space $I$.}
    \label{fig::stinespring}
\end{center}
\end{figure}

\begin{figure}[t]
    \begin{center}
    \includegraphics[width=\columnwidth]{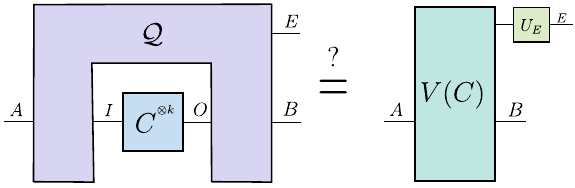}
    \caption{\textbf{Universal quantum purification machines.} A machine that takes as input $k$ copies of an arbitrary quantum channel and outputs one of its possible purifications, either probabilistically or approximately.} 
    \label{fig::purificationmachine}
\end{center}
\end{figure}

A basic structural feature of quantum theory is the ubiquity of purification and dilation: mixed, noisy, or irreversible descriptions can often be realized as reductions of pure, reversible, or higher-dimensional ones~\cite{stinespring1995positive,naimark1943representation,ozawa1984quantum,szHokefalvi1954contractions}. This paradigm arises across many classes of quantum objects. In particular, every quantum state is the marginal of a pure state on a larger Hilbert space, and every quantum channel admits a isometric realization on a larger system and environment whose partial trace reproduces the channel. The purification principle for states and the Stinespring dilation theorem for channels thus exemplify a common idea: mixedness, noise, and irreversibility originate from the neglect of degrees of freedom in a larger reversible description. Beyond its conceptual significance, this paradigm plays an indispensable role in quantum information science, providing a framework for open-system dynamics~\cite{evans1977dilations,attal2006repeated,vom2024finite}, and enabling the design and analysis of quantum protocols in which ancillary degrees of freedom are explicitly modeled~\cite{AcinBae2006,Dupuis2014,Wang23}.

Because of their generality, purifications of states and Stinespring dilations of channels are often treated as canonical representations of the corresponding quantum objects. Moreover, since every quantum state can be identified with a quantum channel that has a trivial, one-dimensional input system, purification of states may be viewed as a special case of Stinespring dilation. Motivated by this identification, we will refer to Stinespring dilations of quantum channels, including the state case, simply as purifications. It is then natural to ask whether such purifications can be obtained systematically---for instance, whether there could exist a universal transformation that, given as input an arbitrary mixed state or noisy channel, outputs one of its possible purifications. A map of this kind would provide a powerful bridge between the operational level of description, where one has access to a quantum state or channel only in the form of a black box, and a purified representation, where the object is embedded into a larger reversible evolution.

However, quantum theory strictly forbids the existence of such universal purification transformations. The reason is intuitively clear: a universal purification map would necessarily act as an inverse to the partial trace. By construction, purifications are obtained from larger systems by discarding auxiliary degrees of freedom; reversing this process would amount to retrieving information and correlations lost to the environment from a local description alone. Such a scenario directly conflicts with thermodynamic principles: the partial trace models an irreversible loss of information, resulting in entropy increase and the emergence of irreversibility. A transformation capable of undoing this operation for arbitrary inputs would therefore amount to a device that systematically decreases entropy, in violation of the second law of thermodynamics.

While intuitive, this limitation of quantum theory has only recently been formalized in a rigorous no-go theorem~\cite{liu2025nouniversal}. In this work, we further explore the limits of quantum purification by formalizing and investigating this task in the probabilistic exact and deterministic approximate settings. In the probabilistic exact setting, we offer a simple proof of impossibility that does not require universality to hold, from which it follows as a corollary that universal purification of quantum states and channels, for any finite number of input copies, cannot succeed with non-zero probability. In the deterministic approximate setting, we focus on the minimum average purification error, defined as the minimum squared Hilbert-Schmidt distance between the output of the transformation and any of the possible purifications of the input, averaged over the ensemble of quantum channels induced by Haar-distributed Stinespring isometries. This framework allows us to analyze the dimension of the environment system as a prior distribution over the possible inputs of transformation. We study various families of purification strategies, analytically computing their attained error explicitly for both quantum states and quantum channels, and comparing their optimality in different regimes of environment dimension. We investigate strategies that always produce a pure output, and show that the optimal ones are precisely those whose outputs are purifications of the fully depolarizing channel. We then characterize the average purification error of this family in terms of the amount of entanglement between the original systems and the environment system of the pure output. Next, we investigate strategies that simply append a state onto the environment system, leaving the input unchanged, computing the minimum average error among all possible appended states and showing that, in the regime of low environment dimension, this strategy performs better than any strategy that produces a pure output. We further characterize the average error attained by a strategy that maps all inputs to the fully depolarizing channel. Next, we describe a many-copy estimation strategy based on the average purification transformation \cite{PelecanosEtAl2025,girardi2025random}, compute a general error upper bound, and finally compare the performance of these different classes of strategies. We conclude by comparing our universal purification task with universal cloning~\cite{wootters1982asingle}, universal unitary controlization~\cite{araujo2014quantum,gavorova2024topological}, and the task of transforming $k$ copies of an arbitrary input into the average of $k$ copies of its possible purifications~\cite{chen2024localtestunitarilyinvariant,chen2025quantumchanneltomographyestimation,tang2025,girardi2025random2,yoshida2025}.

\section{Quantum states, quantum channels, and their purifications}

Quantum channels are characterized by completely positive (CP) trace-preserving (TP) linear maps $\widetilde{C}:\Lcal (\Hcal_I)\to\Lcal (\Hcal_O)$ that transform states acting on an linear input space $\Lcal (\Hcal_I)$ of finite dimension $d_I\coloneqq\text{dim}(\Hcal_I)$, to states acting on a linear output space of finite dimension $d_O\coloneqq\text{dim}(\Hcal_O)$. The action of quantum channels on input quantum states $\rho\in\Lcal (\Hcal_I)$ can be described by their Kraus decomposition, according to $\widetilde{C}(\rho)= \sum_i K_i \rho K_i^* \in \Lcal (\Hcal_O)$, where $K_i:\Hcal_I\to\Hcal_O$, called the Kraus operators, satisfy $K_i\geq0$ and $\sum_i K_i^* K_i=\id$. A relevant case of quantum channel is that of an isometric channel $\widetilde{V}(\rho)=V\rho V^*$, a quantum operation that maps pure states, of the form $\rho=\ketbra{\psi}$, into other pure states, and is described by a single Kraus operator $V:\Hcal_I\to\Hcal_O$ satisfying $V^* V=\id$. Unitary channels are a particular case of isometric channels, in the case where $d_I=d_O=d$, that is, $\Hcal_I\cong \Hcal_O$, and correspond to reversible transformations.

Using the Choi-Jamio\l{}kowski isomorphism, a quantum channel can be equivalently represented by a linear operator $C\in\Lcal (\Hcal_I\otimes\Hcal_O)$ that acts on the joint input-output space, defined as $C\coloneqq\sum_{i,j}\ketbra{i}{j}\otimes\widetilde{C}(\ketbra{i}{j})$, where $\{\ket{i}\}_i$ denotes the computational basis. The action of a quantum channel on an input quantum state is obtained from its Choi operator via the relation $\widetilde{C}(\rho)=\tr_{I}[C\,(\rho^T\otimes\id^O)]$, where $(\cdot)^T$ denotes transposition in the computational basis and $\id^X$ the identity operator on $\Hcal_X$. Isometric channels are characterized by a rank-one Choi operator $\ketbra{V}$, where $\ket{V}\coloneqq\sum_i\ket{i}\otimes V\ket{i}$. We refer to both the map $\widetilde{C}$ and the operator $C$ equivalently as quantum channels.

According to the Stinespring dilation theorem~\cite{stinespring1995positive}, every non-isometric quantum channel $\widetilde{C}:\Lcal (\Hcal_I)\to\Lcal (\Hcal_O)$ can be represented by an isometric channel $\widetilde{V}:\Lcal (\Hcal_I)\to\Lcal (\Hcal_O\otimes \Hcal_E)$, referred to as its purification, with output in a larger space that includes an auxiliary system $\Hcal_E$ representing the environment, followed by a partial trace that discards the auxiliary system, as depicted in Fig.~\ref{fig::stinespring}. Formally, for every quantum channel $\widetilde{C}$ there exists a family of isometries $V(C):\Hcal_I\to\Hcal_O\otimes\Hcal_E$ that satisfy
\begin{equation}\label{eq::defstines1}
    \widetilde{C}(\rho) = \tr_E [V(C)\,\rho\, V(C)^*] \ \ \ \forall \ \rho,
\end{equation}
or, equivalently, in the Choi representation,
\begin{equation}\label{eq::defstines2}
    C = \tr_E \left[\ketbra{V(C)}\right].
\end{equation}
These isometric operators can be constructed from any Kraus decomposition $\{K_i\}_i$ of the channel $\widetilde{C}$ according to 
\begin{equation}\label{eq::stinesiscometry}
    V(C) \coloneqq \sum_i K_i \otimes \ket{\phi_i}, 
\end{equation}
where $\{\ket{\phi_i}\}$ is any orthonormal basis on $\Hcal_E$. For every quantum channel $\widetilde{C}$, there exits a purification with environment dimension $d_E=\rank(C)\leq d_Id_O$. Conversely, every isometric channel $\widetilde{V}:\Lcal (\Hcal_I)\to\Lcal (\Hcal_O\otimes \Hcal_E)$ is the Stinespring dilation of some quantum channel $\widetilde{C}:\Lcal (\Hcal_I)\to\Lcal (\Hcal_O)$.

For a fixed environment space $\Hcal_E$, different choices of Kraus decomposition and of basis for the environment space will lead to different purifications of the same channel, all of which are related to each other via a unitary operator $U_E:\Hcal_E\to\Hcal_E$ acting on the environment. That is, the Choi operators of all possible purifications of a channel $\widetilde{C}$ can be attained from a fixed purification $V(C)$, constructed with some choice of $\{K_i\}_i$ and $\{\ket{\phi_i}\}_i$ in the form of Eq.~\eqref{eq::stinesiscometry}, via
\begin{equation}
    \ket{V'(C)} = \left(\id\otimes U_E\right)\ket{V(C)}.
\end{equation}

As already noted in Sec.~\ref{sec:: Introduction},  quantum states themselves can be regarded as a special case of quantum channels, those that have an input space of dimension $d_I=1$, i.e., a trivial input space. Under this identification, a density operator $\rho$ can be viewed as a particular case of the Choi operator of a quantum channel $C$, with $d_I=1$. In particular, pure states can be represented as isometric channels with $d_I=1$, with $\ket{\psi}$ seen as a particular case of $\ket{V}$. Accordingly, if a mixed state admits the spectral decomposition $\rho=\sum_i p_i \ketbra{\psi_i}{\psi_i}$, its possible purifications  $\ket{\psi(\rho)}\coloneqq\sum_i\sqrt{p_i}\ket{\psi_i}\otimes\ket{\phi_i}$, where again we have $\{\ket{\phi_i}\}$ being an orthonormal basis on $\Hcal_E$, are also related by a unitary operator $U_E$ acting on the environment, via $\ket{\psi'(\rho)} = (\id\otimes U_E)\ket{\psi(\rho)}$. Hence, state purifications $\ket{\psi(\rho)}$ can be seen as Stinespring dilations of a channel  $\ket{V(C)}$ for channels with a trivial input system. This identification will be useful for the interpretation of our results.

\section{Impossibility of universal quantum purification machines}

A quantum purification machine is a hypothetical device that can take an arbitrary quantum channel as input, which may be provided as a black box, and transform it into one of its valid purifications. Such a machine can be modeled by a higher-order transformation, namely a linear map whose inputs are themselves maps such as quantum channels~\cite{chiribella2008quantum,chiribella2008transforming,taranto2025higher}, as depicted in Fig.~\ref{fig::purificationmachine}. Throughout, we describe purification machines as linear maps $\Qcal:\Lcal (\Hcal_I\otimes\Hcal_O)\to\Lcal (\Hcal_A\otimes\Hcal_B\otimes\Hcal_E)$, where $\Hcal_A\cong\Hcal_I$ and $\Hcal_B\cong\Hcal_O$, acting on the Choi operator $C$ of the input quantum channel. $\Qcal$ is a universal purification machine if, for every channel $C$, there exists a (potentially different) unitary $U_E(C)$ such that
\begin{align}\label{eq::determ_exact_machine}
    \Qcal(C) &= (\id\otimes U_E(C))\ketbra{V(C)}(\id\otimes U_E(C)^*) \\
    &=: \ketbra{V_{U_E}(C)}.
\end{align}
Equivalently, for every input channel $C$, $\Qcal(C)$ satisfies $\rank[\Qcal(C)]=1$ and $\tr_E[\Qcal(C)]=C$.

At this stage, we require $\Qcal$ to produce such an output deterministically and exactly for every input channel, while imposing no structural assumptions beyond linearity. This is already enough to rule out the existence of a universal purification machine. Indeed, let $C$ be a non-extremal quantum channel. Since the set of quantum channels is convex, one can write $C=q\,C_1+(1-q)\,C_2$ for some distinct quantum channels $C_1$ and $C_2$, and some $q\in(0,1)$. Then by hypothesis, the output will be $\Qcal(C) = \ketbra{V_{U_E}(C)}$ with some unitary $U_E(C)$. From the linearity of $\Qcal$, we also get that $\Qcal(C)=q\Qcal(C_1)+(1-q)\Qcal(C_2)$, and by hypothesis, that $\Qcal(C)=q\ketbra{V_{U_E}(C_1)}+(1-q)\ketbra{V_{U_E}(C_2)}$ from some unitaries $U_{E}(C_1)$ and $U_{E}(C_2)$. However, this leads to a contradiction, since for all unitaries  $U_E(C)$, $U_E(C_1)$, and $U_E(C_2)$, we have that $\ketbra{V_{U_E}(C)}\neq q\ketbra{V_{U_E}(C_1)}+(1-q)\ketbra{V_{U_E}(C_2)}$, since the l.h.s. has rank $1$ for all $U_E(C)$, while in the r.h.s. has rank greater than $1$ for all $U_E(C_1),U_E(C_2)$. This can be seen from the fact that, by virtue of being purifications of two different channels, $\ket{V_{U_E}(C_1)}$ and $\ket{V_{U_E}(C_2)}$ are linearly independent for all $U_E(C_1),U_E(C_2)$, and the convex combination of two linearly independent rank-$1$ operators with convex weight $q\in(0,1)$ necessarily increases the rank. This contradiction shows that no linear universal purification machine can exist. One way to relax the machine's task is to give it access to more calls of the input channels. In this scenario, a linear machine that has access to infinitely many calls of the arbitrary input channels can perform process tomography of the input, completely characterizing it, and subsequently prepare one of its valid purifications as output. Therefore, the interesting, non-trivial question is how well a machine can perform the purification task when only a finite number of uses of the input channel are available, and hence only partial information about it can be extracted. In the following we consider the probabilistic exact and the deterministic approximate versions of this purification problem.

\section{Probabilistic exact quantum purification machines}

In the many-copy, probabilistic exact case, we consider a machine $\Qcal$ that corresponds to a linear map and transforms a finite number of copies of an input into one of its valid purifications with some non-zero probability. It is known that such a universal purification machine cannot be described by a positive linear map~\cite{liu2025nouniversal}. Here, we provide a simple result, focusing on the case of quantum states ($d_I=d_A=1$), showing that it is not necessary to require the universality of the machine to arrive at this conclusion---indeed, it is only necessary to impose that the machine probabilistically purifies two specific input quantum states, a rank-$1$ state and a rank-$2$ state, with some non-zero probability.

\begin{theorem}[No-go theorem for probabilistic exact non-universal many-copy purification]
    For all finite $k$, there is no linear positive map  $\Qcal:\Lcal (\Hcal_O^{\otimes k})\to\Lcal (\Hcal_B\otimes\Hcal_E)$ such that, for two rank-$1$ quantum states $\ketbra{\psi_0},\ketbra{\psi_1}\in\Lcal (\Hcal_O)$, with $\ketbra{\psi_0}\neq\ketbra{\psi_1}$, $\Qcal$ satisfies
    \begin{equation}
       \Qcal[\ketbra{\psi_0}^{\otimes k}] = p_0 \ketbra{\psi_0}\otimes\ketbra{\phi_0} 
    \end{equation}
    and
    \begin{equation}
        \Qcal[(q\ketbra{\psi_0}+(1-q)\ketbra{\psi_1})^{\otimes k}] = p_{01} \ketbra{\phi_{01}},
    \end{equation}
    where $\tr_E(\ketbra{\phi_{01}}) = q\,\ketbra{\psi_0}+(1-q)\ketbra{\psi_1}$, and $q\in(0,1)$, with
    \begin{equation}
        p_0>0 \quad \text{and} \quad p_{01}>0.
    \end{equation}
\end{theorem}

\vspace*{11pt}

\begin{proof}
By the linearity of $\Qcal$, it follows that 
\begin{align}
    &\Qcal[(q\,\ketbra{\psi_0}+(1-q)\ketbra{\psi_1})^{\otimes k}] \nonumber \\
    &= q^k \Qcal(\ketbra{\psi_0}^{\otimes k}) \nonumber \\
    &+ \sum_{i=1}^{k} q^{k-i}(1-q)^{i}\Qcal(\ketbra{\psi_0}^{\otimes(k-i)}\otimes\ketbra{\psi_1}^{\otimes i}) \\
    &= q^k p_0 \ketbra{\psi_0}\otimes\ketbra{\phi_0} \nonumber \\
    &+ \sum_{i=1}^{k} q^{k-i}(1-q)^{i}\Qcal(\ketbra{\psi_0}^{\otimes(k-i)}\otimes\ketbra{\psi_1}^{\otimes i})
\end{align}
which, by assumption, must equal $p_{01} \ketbra{\phi_{01}}$. Now assume, by contradiction, that $\Qcal$ is a positive map. Then, every term in the sum above is positive semidefinite, and therefore
\begin{equation}
    p_{01} \ketbra{\phi_{01}} \geq q^k p_0 \ketbra{\psi_0}\otimes\ketbra{\phi_0}.
\end{equation}
Since $p_{01}>0$, $p_0>0$ and $q\in(0,1)$, both sides are non-zero rank-$1$ operators, which implies that $\ket{\phi_{01}}$ must be proportional to $\ket{\psi_0}\otimes\ket{\phi_0}$. Consequently,
\begin{equation}
    \tr_E(\ketbra{\phi_{01}}) = \ketbra{\psi_0},
\end{equation}
which contradicts the assumption that  $\tr_E(\ketbra{\phi_{01}}) = q\,\ketbra{\psi_0}+(1-q)\ketbra{\psi_1}$ for some $q\in(0,1)$. Hence, no such linear positive map $\Qcal$ can exist.
\end{proof}

The case of quantum channels and universal purification is hence obtained as corollary.
Let $\Qcal_p$ be a linear higher-order transformation $\Qcal_p:\Lcal (\Hcal_I^{\otimes k}\otimes\Hcal_O^{\otimes k})\to\Lcal (\Hcal_A\otimes\Hcal_B\otimes\Hcal_E)$ that corresponds to a universal probabilistic purification machine. That is, $\Qcal_p$  outputs exactly one of the possible purifications of the input channel with certain non-negative probability $p$, and with probability $1-p$, fails. For a fixed input and output space dimension $d_I,d_O$, and a fixed auxiliary space dimension $d_E$, we are then interested in the maximal probability of success $p_s$ of such a protocol, which is defined as 
\vspace*{-1em}
\begin{widetext}
\begin{equation}\label{eq::max_psucc_def}
     p_s(d_I,d_O,d_E;k) \coloneqq \max_{\Qcal_p} \left\{ p\ \Big|\ \forall\ C, \exists\ U_E(C);\ \Qcal_p\left(C^{\otimes k}\right) = p \, \left(\id\otimes U_E(C)\right)\ketbra{V(C)}\left(\id\otimes U_E(C)^*\right) \right\}.
\end{equation}
\end{widetext}

We then obtain the following result concerning the general performance of universal probabilistic many-copy quantum channel purification machines:

\begin{corollary}[No-go theorem for many-copy probabilistic exact universal purification]\label{thm::probkcopy} 
    Let $C\in\Lcal (\Hcal_I\otimes\Hcal_O)$ be the Choi operator of a quantum channel, $\ketbra{V_{U_E}(C)}$ be Choi operator of one of its possible Stinespring dilations, and let $\Qcal_p:\Lcal (\Hcal_I^{\otimes k}\otimes\Hcal_O^{\otimes k})\to\Lcal (\Hcal_A\otimes\Hcal_B\otimes\Hcal_E)$ be a universal probabilistic purification machine, described by a linear positive map, acting on $k$ copies of $C$. Then, the maximal probability of success $p_s(d_I,d_O,d_E;k)$, defined in Eq.~\eqref{eq::max_psucc_def}, of such a transformation satisfies
    \begin{equation}
        p_s(d_I,d_O,d_E;k) = 0,
    \end{equation}
    for all $d_I$, $d_O$, $d_E$, and $k$.
\end{corollary}

This result showcases a striking prohibition imposed by quantum theory: exact universal purification is completely forbidden, even in the probabilistic case. The simple requirements of linearity and positivity of the map---minimal requirements of a quantum transformation---are sufficient to show that, no matter how many copies of the input are provided, probabilistic exact purification is impossible.

\section{Deterministic approximate quantum purification machines}

In the deterministic approximate case, we admit a purification machine, described by a higher-order transformation $\Qcal_\epsilon:\Lcal (\Hcal_I^{\otimes k}\otimes\Hcal_O^{\otimes k})\to\Lcal (\Hcal_A\otimes\Hcal_B\otimes\Hcal_E)$, which is not required to output an exact purification of the input channel but rather an approximation of it. We require $\Qcal_\epsilon$ to be a valid general $k$-slot superchannel, that is, a higher-order transformation that maps $k$ channels into one channel, even when acting on part of its inputs~\cite{chiribella2008quantum,chiribella2008transforming,taranto2025higher}. Formally, $\Qcal_\epsilon$ must be completely CPTP preserving. This guarantees that, for all inputs, the approximate purification machine outputs a valid quantum channel. General superchannels have been extensively studied in the literature of higher-order transformations~\cite{chiribella2008quantum,chiribella2008transforming,taranto2025higher}. Such maps, while encompassing all circuit-implementable transformations, are in general not required to respect strict causality constraints, admitting transformations that may involve an indefinite causal order~\cite{chiribella2013quantum,araujo2017purification}.

For a fixed dimensions $d_I,d_O, d_E$ and a fixed number of copies $k$, given an error function $f$, which quantifies for each input $C$ the deviation of the machine output from an exact purification, we define the minimum average $k$-copy purification error $\epsilon(d_I,d_O,d_E;k)$ as
\vspace*{-1em}
\begin{widetext}
\begin{equation}\label{eq::avg_error_def}
    \epsilon(d_I,d_O,d_E;k)\coloneqq\min_{\Qcal_\epsilon}\left\{\EE_{C}\left\{\min_{U_E}\Big\{ 
    f\left(\Qcal_\epsilon\left(C^{\otimes k}\right), \ketbra{V_{U_E}(C)}\right)\Big\}\right\}\right\},
\end{equation}
\end{widetext}
\vspace*{-1em}
for all $d_I,d_E\geq1$ and $d_O\geq2$.

To quantify this purification error, we choose as figure of merit the squared Hilbert-Schmidt distance $f(A,B)=\norm{A-B}_2^2$. This quantity is interesting in the current scenario because it captures several of its defining characteristics. In the present setting, we have that 
\begin{align}
\begin{split}
    &\norm{\Qcal_\epsilon\left(C^{\otimes k}\right) - \ketbra{V_{U_E}(C)}}^2_2 = \\
    & \quad = \tr\left[\Qcal_\epsilon\left(C^{\otimes k}\right)^2\right] + \tr\left[\ketbra{V_{U_E}(C)}^2\right] + \\
    & \quad\quad - 2\tr\left[\Qcal_\epsilon\left(C^{\otimes k}\right)\ketbra{V_{U_E}(C)}\right].
\end{split}
\end{align}
The first term, $\tr\,[\Qcal_\epsilon\left(C^{\otimes k}\right)^2]$, quantifies the purity of the output of the machine. Notice that here we do not demand the output of the machine to be pure---since this is not necessarily the best approximate strategy---but rather consider completely general machines with no restrictions on their output. The second term, $\tr\,[\ketbra{V_{U_E}(C)}^2]=d_I^2$, is fixed by the fact that the target is a pure Choi operator. The last term, $\tr[\Qcal_\epsilon\left(C^{\otimes k}\right)\ketbra{V_{U_E}(C)}]=\abs{\bra{V_{U_E}(C)}\Qcal_\epsilon\left(C^{\otimes k}\right)\ket{V_{U_E}(C)}}$, is proportional to the fidelity between the output of the machine and the target of the transformation. The squared Hilbert-Schmidt norm then captures the trade-off between output purity and alignment with a valid purification of the input.

To compute the average over quantum channels $\EE_{C}$, we sample each input channel by drawing a Haar-random isometry $V(C)$ for each fixed auxiliary system dimension $d_E$~\cite{kukulski2021generating}, knowing that $C=\tr_E\ketbra{V(C)}$. See App.~\ref{ap:: reviewrandomchannels} for more details. Hence,
\vspace*{-1em}
\begin{widetext}
\begin{align}
    &\epsilon(d_I,d_O,d_E;k) = \min_{\Qcal_\epsilon}\left\{\EE_{C}\left\{\min_{U_E}\Big\{\norm{\Qcal_\epsilon\left(C^{\otimes k}\right) - \ketbra{V_{U_E}(C)}}^2_2 \Big\}\right\}\right\} \\
    &= \min_{\Qcal_\epsilon}\left\{\EE_{V}\left\{\min_{U_E}\Big\{ \norm{\Qcal_\epsilon\left(\tr_E(\ketbra{V})^{\otimes k}\right) - (\id\otimes U_E)\ketbra{V}(\id\otimes U_E^*)}^2_2 \Big\}\right\}\right\} \\
    &= d_I^2 + \min_{\Qcal_\epsilon}\left\{\EE_{V}\left\{ 
    \tr[\Qcal_\epsilon\left(\tr_E(\ketbra{V})^{\otimes k}\right)^2] - 2 \max_{U_E}\Big\{\tr\Big[\Qcal_\epsilon\left(
    \tr_E(\ketbra{V})^{\otimes k}\right)(\id\otimes U_E)\ketbra{V}(\id\otimes U_E^*)\Big]
    \Big\}\right\}\right\}. \label{eq::error_expanded}
\end{align}
\end{widetext}
\vspace*{-1em}

We now discuss a couple of noteworthy remarks about the above error expression. Since the average is taken over Haar random isometries $V:\Hcal_I\to\Hcal_O\otimes\Hcal_E$, the dimension $d_E$ of the environment affects the induced distribution of the channels $C$. In particular, for $d_E=d_Id_O$, this construction induces a uniform distribution over Choi operators of quantum channels $C$~\cite{kukulski2021generating,Collins_Nechita_2016}. For $d_E=1$, the distribution is supported on isometric channels and assigns zero probability to non-isometric ones. Finally, in the limit $\lim d_E\to\infty$, the induced channel concentrates, with probability approaching one, around the fully depolarizing channel. This allows us to interpret the parameter $d_E$ as a \textit{prior distribution} over the input channels $C$ that are given to the machine. Since $d_I,d_O,d_E$ and $k$ are parameters given to the machine, while it remains agnostic with respect to which input channel it will receive, the machine can be programmed to adapt its purification strategy depending on the information of the prior distribution over $C$ provided in the form of the parameter $d_E$.

In the regime of $d_E\ll d_Id_O$, the machine can expect, with high probability, to receive an input channel that is close to isometric. Hence, on the specific case of $d_E=1$, a strategy that attains zero error exists: since the input is known to be an isometric channel, the optimal strategy is to simply output the input channel, unchanged. In the regime where $d_E\approx d_Id_O$, the machine expects to receive a channel $C$ sampled approximately uniformly. In the case of $d_E=d_Id_O$, it is expected that the error will be maximal, as this is the scenario where the prior is least informative about the specific input channel. Finally, in the regime of $d_E\gg d_Id_O$, the machine can expect to receive an input, with high probability, close to the fully depolarizing channel $\widetilde{D}(\rho)=\tr(\rho)\id/d_I$, which maps all quantum states to the maximally mixed state. Hence, in the limit of $d_E\to\infty$, an optimal strategy that attains zero error also exists: by knowing that the input channel is the fully depolarizing channel, the machine can simply prepare one of its possible purifications. The same reasoning holds for states, which again is the case recovered when $d_I=1$, and an analogous behavior of the error predicted. Finally, for arbitrary $d_I$, $d_O$, and $d_E$,  there exists a strategy that becomes exact in the limit $k\to\infty$. In this limit, the machine can perform exact tomography of the input channel and simply prepare one of its possible purifications as output. However, this strategy necessarily incurs a nonzero error for all finite $k$, including the regimes of $d_E=1$ and $d_E\to\infty$, where an optimal strategy attaining zero error exists.
In Fig.~\ref{fig::error}, we plot the expected scaling of $\epsilon$ with $d_E$ described above.

\begin{figure}[t]
    \begin{center}
    \includegraphics[width=\columnwidth]{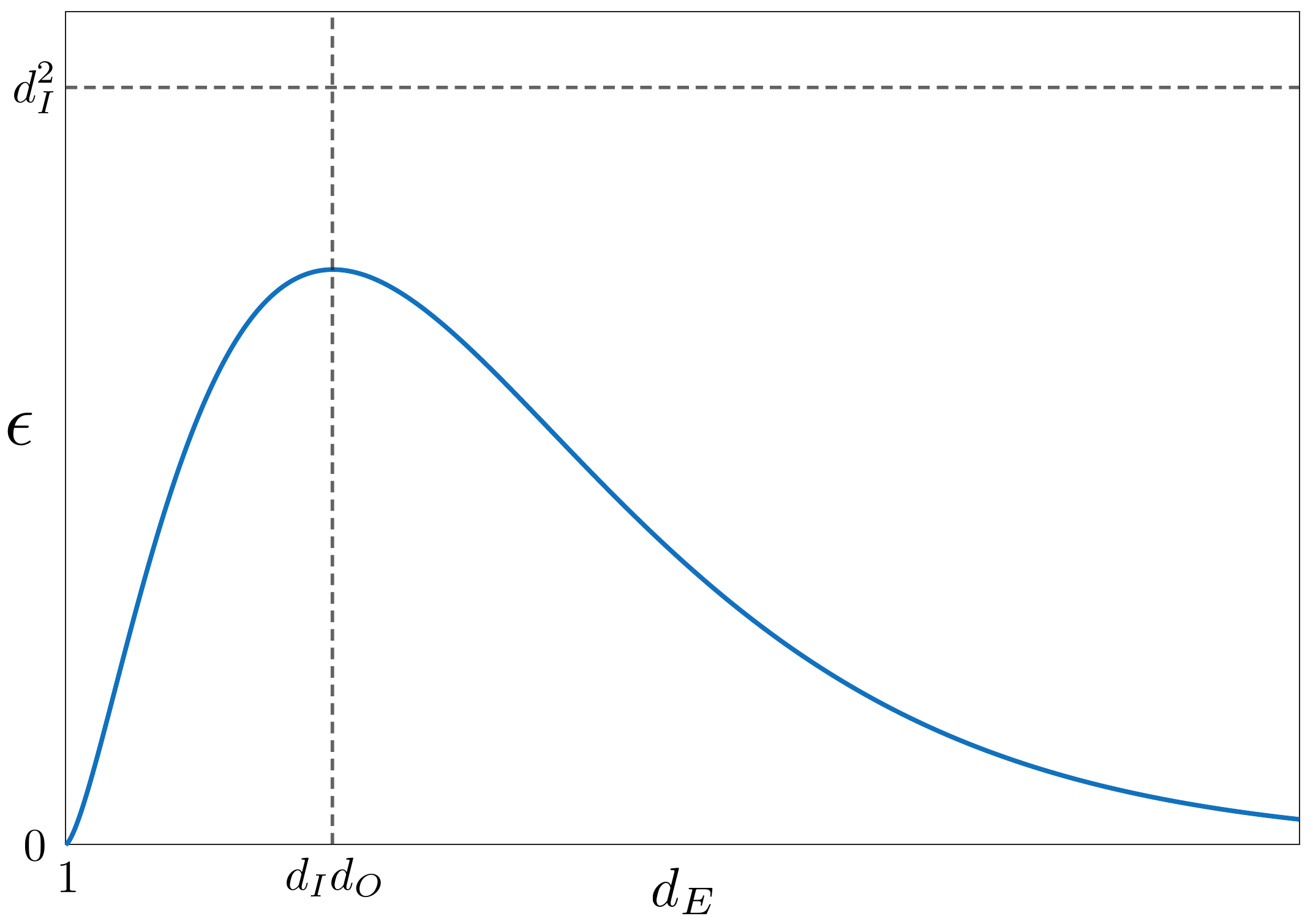}
    \caption{\textbf{Expected scaling of the minimum average error.} Expected scaling of the minimum average error $\epsilon(d_I,d_O,d_E;k)$ in Eq.~\eqref{eq::avg_error_def} as a function of the prior distribution given by the environment dimension $d_E$.} 
\label{fig::error}
\end{center}
\end{figure}

We now compare the performance of different families of physically motivated strategies that the machine can adopt. 

\subsection{Pure-output strategy}\label{sec::pureoutput}

The first family of strategies we consider consists of machines that always output a pure Choi operator, i.e., an isometric channel. This case was considered in Ref.~\cite{liu2025nouniversal} for state purification, where it was shown that if a universal purification machine described by a linear and positive map is constrained to produce a pure output, then it can only produce a fixed pure output, independently of the input state and the number of copies available. We therefore analyze here the performance of such a purification machine, which applies a kind of ``discard-and-reprepare'' strategy that ignores the input channel $C$ and prepares as output a fixed isometric channel. As shown below, this strategy becomes optimal in the limit $d_E\to\infty$, where the input concentrates on the fully depolarizing channel.

Let $\Qcal_{\ket{W}}$ denote a universal approximate purification machine that deterministically outputs a fixed isometric channel $\ketbra{W}$, according to
\begin{equation}
    \Qcal_{\ket{W}}(C)\coloneqq \tr(C)\frac{\ketbra{W}}{d_I} = \ketbra{W}
\end{equation}
for all $C$. The corresponding optimal error within this family is
\begin{equation}\label{eq::def_epsilon_pure}
    \epsilon_\mathrm{pure}(d_I,d_O,d_E) \coloneqq \min_{\ket{W}} \epsilon_{\ket{W}}(d_I,d_O,d_E),
\end{equation}
where
\begin{align}
    \epsilon_{\ket{W}}&(d_I,d_O,d_E) \coloneqq \nonumber \\
    & \EE_C \min_{U_E} \norm{\Qcal_{\ket{W}}(C)-\ketbra{V_{U_E}(C)}}^2_2 \\
    & = \EE_C \min_{U_E}\norm{\ketbra{W}-\ketbra{V_{U_E}(C)}}^2_2.
\end{align}
Notice that $\epsilon(d_I,d_O,d_E;k)\leq\epsilon_\mathrm{pure}(d_I,d_O,d_E)\leq\epsilon_{\ket{W}}(d_I,d_O,d_E)$ for all $d_I,d_I,d_E,k$ and $\ket{W}$.

As we show in App.~\ref{app::pureoutput}, Lemma~\ref{lem::error_pureoutput},
\begin{align}
    \epsilon_{\ket{W}}&(d_I,d_O,d_E) = \nonumber \\
    & 2d_I^2 - 2 d_I^2 \EE_C\left\{ F\left(\frac{C}{d_I},\frac{\tr_E\ketbra{W}}{d_I}\right)\right\} \\
    &= 2d_I^2 - 2\EE_C\left\{\tr(\sqrt{C}\sqrt{\tr_E\ketbra{W}})^2\right\},
\end{align}
where $F(\rho,\sigma)=\tr(\sqrt{\rho}\sqrt{\sigma})^2$ is state fidelity.

As we prove in App.~\ref{app::pureoutput}, the performance of different $\Qcal_{\ket{W}}$ is determined by the amount of entanglement in the $AB|E$ bipartition of $\ket{W}$. Let $\ket{\Omega}$ be a maximally entangled across the $AB|E$ cut, such that $\tr_E\ketbra{\Omega}=\id/d_O$, i.e., $\ket{\Omega}$ is a purification of the fully depolarizing channel $\widetilde{D}$, with Choi operator $D=\id/d_O$. Now let $\ket{\Upsilon}\otimes\ket{\psi}$ be separable across the $AB|E$ cut, so that it represents a purification of an isometric channel $\ketbra{\Upsilon}$. Then, we prove that $\epsilon_{\ket{W}}$ can be tightly upper and lower bounded according to
\begin{equation}
    \epsilon_{\ket{\Omega}}(d_I,d_O,d_E) \leq \epsilon_{\ket{W}}(d_I,d_O,d_E) \leq \epsilon_{\ket{\Upsilon}\otimes\ket{\psi}}(d_I,d_O,d_E),
\end{equation}
where, as shown in App.~\ref{app::pureoutput}, Lemma~\ref{lem::upperb_pureoutput}, we compute that
\begin{equation}
    \epsilon_{\ket{\Upsilon}\otimes\ket{\psi}}(d_I,d_O,d_E) = 2\left(d_I^2 - \frac{d_I}{d_O}\right),
\end{equation}
which is independent of $\ket{\Upsilon}$, $\ket{\psi}$, and $d_E$, and, as proved in App.~\ref{app::pureoutput}, Lemma~\ref{lem::lowerb_pureoutput}, we compute that
\begin{equation}
    \epsilon_{\ket{\Omega}}(d_I,d_O,d_E) = 2d_I^2 - \frac{2}{d_O}\EE_C\left\{\tr(\sqrt{C})^2\right\},
\end{equation}
which is independent of $\ket{\Omega}$. 

This implies that the optimal pure-output strategy is the one that always outputs a purification $\ketbra{\Omega}$ of the fully depolarizing channel. Hence, as we prove in App.~\ref{app::pureoutput},

\begin{theorem}[Minimum error of pure-output strategies]\label{thm::pure-output}
    The minimum average purification error $\epsilon_\mathrm{pure}(d_I,d_O,d_E)$ over all pure-output strategies $\Qcal_{\ket{W}}(C)=\ketbra{W}$, as defined in Eq.~\eqref{eq::def_epsilon_pure}, is  
    \begin{equation}\label{eq::epsilon_pureoutput_solution}
        \epsilon_\mathrm{pure}(d_I,d_O,d_E) = 2d_I^2 - \frac{2}{d_O}\EE_C\left\{\tr(\sqrt{C})^2\right\}.
    \end{equation}
\end{theorem}
\vspace*{11pt}

In particular, no partially entangled choice of $\ket{W}$ across $AB|E$ can outperform a maximally entangled $\ket{\Omega}$ while separable outputs $\ket{\Upsilon}\otimes\ket{\psi}$ are the worst within this pure-output family.

In App.~\ref{ap:: reviewrandomchannels}, Lemma~\ref{lemma:: asymptotics trace sqrt of C}, in analyze the behavior of the quantity $\EE_C\{\tr\,(\sqrt{C})^2\}$ in asymptotic regimes of $d_Id_O$, for different ranges of $d_E$.

We now discuss the error $\epsilon_\mathrm{pure}(d_I,d_O,d_E)$ in the limiting regimes of $d_E$. For $d_E=1$, we sample only isometric channels $C=\ketbra{V}$, given that the environment system has a trivial dimension. Since $\ketbra{V}$ is the Choi operator of an isometric channel, it is proportional to a projector, with trace equal to $d_I$. Hence, for every $\ketbra{V}$, $\tr\,(\sqrt{\ketbra{V}})^2=d_I$ for every $\ketbra{V}$. Consequently, in the case of $d_E=1$, Eq.~\eqref{eq::epsilon_pureoutput_solution} yields
\begin{equation}
     \epsilon_\mathrm{pure}(d_I,d_O,d_E) = 2\left(d_I^2 - \frac{d_I}{d_O}\right).
\end{equation}
Therefore, this strategy performs badly in the case where the machine's input is already pure. In fact, as we will see in the following, this strategy performs even worse then one which maps all inputs to the actual maximally depolarizing channel.

In the limit where $d_E\to\infty$, however, the fixed pure-output machine performs optimally, since the input is the maximally depolarizing channel and the strategy of outputting $\ketbra{\Omega}$, one purification of the depolarizing channel, attains zero error, according to
\begin{align}
    &\epsilon_\mathrm{pure}(d_I,d_O,d_E\to\infty)= 2d_I^2 - \frac{2}{d_O}\tr(\sqrt{\frac{\id}{d_O}})^2 =0.
\end{align}

\subsection{Append-environment strategy}\label{sec::appendenv}

The append-environment strategy is a family of ``do-nothing'' strategies in which the machine leaves the input $C$ unchanged and appends a fixed state $\rho_E$ onto the environment. This is the optimal strategy when $d_E=1$, and the input quantum channel is guaranteed to be already a pure, isometric channel. 

Such strategies are defined by 
\begin{equation}
    \Qcal_{\mathrm{app-}\rho_E}(C) \coloneqq C\otimes\rho_{E}
\end{equation}
for all $C$. The error within this family is
\begin{equation}\label{eq::def_appendenv}
    \epsilon_\mathrm{app}(d_I,d_O,d_E)\coloneqq \min_{\rho_E} \epsilon_{\mathrm{app-}\rho_E}(d_I,d_O,d_E)
\end{equation}
where
\begin{align}
    \epsilon_{\mathrm{app-}\rho_E}&(d_I,d_O,d_E)\coloneqq  \nonumber \\ 
    &\EE_C \min_{U_E}\norm{\Qcal_{\mathrm{app-}\rho_E}(C)-\ketbra{V_{U_E}(C)}}^2_2 \\
    &= \EE_C \min_{U_E}\norm{C\otimes\rho_{E}-\ketbra{V_{U_E}(C)}}^2_2
\end{align}
which satisfies $\epsilon(d_I,d_O,d_E;k)\leq\epsilon_{\mathrm{app-}\rho_E}(d_I,d_O,d_E)$ for all $d_I,d_I,d_E,k$ and $\rho_E$.

As we prove in App.~\ref{app::appendenvironment},

\begin{theorem}[Minimum error of append-environment strategies]\label{thm::append-env}
    The minimum average purification error $\epsilon_\mathrm{app}(d_I,d_O,d_E)$ over all append-environment strategies $\Qcal_{\mathrm{app-}\rho_E}(C) = C\otimes\rho_{E}$, as defined in Eq.~\eqref{eq::def_appendenv}, is
    \begin{align}
    \begin{split}\label{eq::epsilon_app_solution}
        &\epsilon_\mathrm{app}(d_I,d_O,d_E) \\
        &= d_I^2-\frac{d_O^2d_E^2-1}{d_Id_O(d_E^2-1)+d_I^2d_E(d_O^2-1)} \sum_{i}(\mathbb{E}_C\!\{(c_i^\downarrow)^2\})^2,
    \end{split}
    \end{align}
    where $\{c_i^\downarrow\}_i$ are the eigenvalues of $C$ in non-increasing order.
\end{theorem}
\vspace*{11pt}

Notice that the quantity that multiplies $\sum_{i}(\mathbb{E}_C\{(c_i^\downarrow)^2\})^2$ in Eq.~\eqref{eq::epsilon_app_solution} is equivalent to $1/\EE_C\{\tr(C^2)\}$, the inverse of the average purity of $C$, which as calculated explicitly in App.~\ref{ap:: reviewrandomchannels}, Lemma~\ref{lemma:: average purity of C}, is
\begin{equation}\label{eq::avg_tr_Csquared}
    \EE_C\{\tr(C^2)\} = \frac{d_Id_O(d_E^2-1)+d_I^2d_E(d_O^2-1)}{d_O^2d_E^2-1}.
\end{equation}
Notice that this quantity is monotonically decreasing in $d_E$ and satisfies
\begin{equation}
    \frac{d_I}{d_O} \leq \EE_C\{\tr(C^2)\} \leq d_I^2,
\end{equation}
for all $d_I,d_O$, and $d_E$.

The case of quantum states is recovered by taking $d_I=1$ in Eq.~\eqref{eq::epsilon_app_solution}.

As shown in App.~\ref{app::appendenvironment}, the minimal error of append-environment strategies obeys
\begin{align}\label{eq::epsilon_app_inequality}
     d_I^2 - \EE_C\{\tr(C^2)\} \leq \epsilon_\mathrm{app}(d_I,d_O,d_E) \leq d_I^2 - \frac{1}{d_E}\EE_C\{\tr(C^2)\},
\end{align}
The upper bound on $\epsilon_\mathrm{app}$, in the right-most side of Eq.~\eqref{eq::epsilon_app_inequality}, is attained in the case where the appended environment state in maximally mixed, i.e., $\rho_E=\id/d_E$. The case of pure appended environment states is also discussed in App.~\ref{app::appendenvironment}, Prop.~\ref{prop::app_purestate},in which case, the error does not in general compare to the upper bound achieved by appending a maximally mixed state.

We now examine the error $\epsilon_\mathrm{app}(d_I,d_O,d_E)$ in Eq.~\eqref{eq::epsilon_app_solution} in different regimes of $d_E$. In the case of $d_E=1$, we have that $C$ is always isometric, and hence has a single non-zero eigenvalue which equals $d_I$. In this case, $\sum_{i}(\mathbb{E}_C\!\{(c_i^\downarrow)^2\})^2=d_I^4$. For $d_E=1$ we also have that $\EE_C\{\tr(C^2)\}=d_I^2$. Therefore, for trivial-dimension environment system Eq.~\eqref{eq::epsilon_app_solution} yields
\begin{equation}
    \epsilon_{\mathrm{app}}(d_I,d_O,d_E=1) = 0.
\end{equation}
This shows that this strategy is optimal when only isometric channels are given as input. Furthermore, in the limit of $d_E\to\infty$, we have that $C$ is the fully depolarizing channel, and hence has all $d_Id_O$ eigenvalues equal to $1/d_O$. In this case, $\sum_{i}(\mathbb{E}_C\!\{(c_i^\downarrow)^2\})^2=d_I/d_O^3$. For $\lim d_E\to \infty$ we also have that $\EE_C\{\tr(C^2)\}=d_I/d_O$. Therefore, here Eq.~\eqref{eq::epsilon_app_solution} yields

\begin{equation}    
    \epsilon_{\mathrm{app}}(d_I,d_O,d_E\to\infty) = d_I^2-\frac{1}{d_O^2}.
\end{equation}

\subsection{Map-to-depolarizing strategy}\label{sec::maptodep}

This strategy is also of the ``discard-and-reprepare'' type, however, contrarily to the strategies in Sec.~\ref{sec::pureoutput}, it outputs the fully depolarizing channel $\widetilde{D}(\rho)=\trace(\rho){\id}/{d_I}$, with Choi operator $D=\id/{d_O}$, which is the barycenter of the space of CPTP maps. Hence, the quantum machine $\Qcal_\mathrm{dep}$ that implements this strategy is 
\begin{equation}
    \Qcal_\mathrm{dep}(C)\coloneqq\tr(C)\frac{\id}{d_Id_Od_E} = \frac{\id}{d_Od_E}
\end{equation}
for all $C$. The associated error
\begin{align}
\begin{split}\label{eq::def_maptodepol}
    \epsilon_\mathrm{dep}&(d_I,d_O,d_E)\coloneqq \\
    &\EE_C \min_{U_E}\norm{\Qcal_\mathrm{dep}(C)-\ketbra{V_{U_E}(C)}}^2_2,
\end{split}
\end{align}
satisfying $\epsilon(d_I,d_O,d_E;k)\leq\epsilon_\mathrm{dep}(d_I,d_O,d_E)$ for all $d_I,d_I,d_E,k$, can be computed exactly, as we show in App.~\ref{app::maptodepol}:

\begin{theorem}[Error of the map-to-depolarizing strategy]\label{thm::maptodepol}
    The average purification error $\epsilon_\mathrm{dep}(d_I,d_O,d_E)$ of the map-to-depolarizing strategy $\Qcal_\mathrm{dep}(C) = \id/d_Od_E$, as defined in Eq.~\eqref{eq::def_maptodepol}, is
    \begin{equation}\label{eq::epsilon_maptodep_solution}
        \epsilon_\mathrm{dep}(d_I,d_O,d_E) = d_I^2 - \frac{d_I}{d_Od_E}.
    \end{equation}
\end{theorem}
\vspace*{11pt}

In different regimes of $d_E$, this quantity amounts to
\begin{align}
    \epsilon_\mathrm{dep}(d_I,d_O,d_E=1) &= d_I^2 - \frac{d_I}{d_O} \\
    \epsilon_\mathrm{dep}(d_I,d_O,d_E=d_Id_O) &= d_I^2 - \frac{1}{d_O^2} \\
    \epsilon_\mathrm{dep}(d_I,d_O,d_E\to\infty) &= d_I^2.
\end{align}

The case of quantum states is obtained by taking $d_I=1$ in Eq.~\eqref{eq::epsilon_maptodep_solution}.

\subsection{General error upper bound}\label{sec::genupperbound}

A general error upper bound can be computed by replacing, for each input channel, the minimization over environment unitaries by an average over them, as in
\begin{align}
    &\epsilon(d_I,d_O,d_E;k) = \nonumber \\
    &\min_\Qcal \EE_C \min_{U_E} \norm{\Qcal(C^{\otimes k}) - \ketbra{V_{U_E}(C)}}^2_2 \\
    &\leq \min_\Qcal \EE_C \EE_{U_E} \norm{\Qcal(C^{\otimes k}) - \ketbra{V_{U_E}(C)}}^2_2 \\
    &=: \epsilon_{\mathrm{avg-}U_E}(d_I,d_O,d_E), \label{eq::def_genuppbound}
\end{align}
which is computed in App.~\ref{app::genuperbound} to be exactly

\begin{theorem}[General error upper bound]\label{thm::genupperbound}
    The upper bound $\epsilon_{\mathrm{avg-}U_E}(d_I,d_O,d_E)$ for the minimum average purification error, as defined in Eq.~\eqref{eq::def_genuppbound}, is 
    \begin{align}
    \begin{split}\label{eq::epsilon_avgUE_solution}
        \epsilon_{\mathrm{avg-}U_E}&(d_I,d_O,d_E) \\
        &= d_I^2 - \frac{1}{d_E} \frac{d_I^2d_E(d_O^2-1)+d_Id_O(d_E^2-1)}{d_E^2d_O^2-1}
    \end{split}    
    \end{align}
    and is independent of $k$.
\end{theorem}
\vspace*{11pt}

Notice that the second term in the above expression is equal to $\EE_C\{\tr(C^2)\}/d_E$.

In the extremal regimes of $d_E$, we have that this general upper bound yields
\begin{equation}
    \epsilon_{\mathrm{avg-}U_E}(d_I,d_O,d_E=1) = 0,
\end{equation}
attaining the optimal value for $d_E=1$, and
\begin{equation}    
    \epsilon_{\mathrm{avg-}U_E}(d_I,d_O,d_E\to\infty) = d_I^2.
\end{equation} 
The case of quantum states is recovered by taking $d_I=1$ in Eq.~\eqref{eq::epsilon_avgUE_solution}.

Notice that this is, in general, different from the problem considered in Refs.~\cite{chen2024localtestunitarilyinvariant,chen2025quantumchanneltomographyestimation,tang2025,girardi2025random2,yoshida2025}, where the goal of the machine is not to output one of the possible purifications of the input, but rather to output the average over all possible purifications. Namely,
\begin{align}
\begin{split}
    &\min_\Qcal \EE_C \EE_{U_E} \norm{\Qcal(C^{\otimes k}) - \ketbra{V_{U_E}(C)}}^2_2 \\
    &\hspace*{7em} \neq \\
    &\min_\Qcal \EE_C  \norm{\Qcal(C^{\otimes k}) - \EE_{U_E} \ketbra{V_{U_E}(C)}}^2_2,
\end{split}    
\end{align}
since the square of the Hilbert-Schmidt distance is a nonlinear function.
The latter case trivially reduces to 
\begin{align}
\begin{split}
    &\min_\Qcal \EE_C  \norm{\Qcal(C^{\otimes k}) - \EE_{U_E} \ketbra{V_{U_E}(C)}}^2_2 \\
    &= \min_\Qcal \EE_C  \norm{\Qcal(C^{\otimes k}) - C\otimes\frac{\id}{d_E}}^2_2 = 0,
\end{split}
\end{align}
by taking $\Qcal(C^{\otimes k})=C\otimes\id/d_E$, which is not the case for $\epsilon_{\mathrm{avg-}U_E}(d_I,d_O,d_E)$ in general.

\subsection{Estimation-based many-copy strategy}\label{sec::estimation}
All of the previous upper bounds are achieved by strategies that are effectively $1$-copy strategies. That is, considering more copies of the input does not improve the attained error for any of the strategies analyzed so far. We now turn to a genuinely many-copy strategy whose performance depends on $k$, namely an estimation-based protocol. In the limit $k\to\infty$, the unknown input can be characterized perfectly via tomography, and the machine can therefore prepare one of its purifications, attaining zero error. At the same time, for every finite $k$, the error remains strictly positive, even in the case where $d_E=1$ and the input channel is guaranteed to be isometric.

Given \(k\) copies of
the normalized Choi state \(C/d_I\), the machine first applies the mixed-state
tomography procedure of Ref.~\cite[Theorem~1.3]{PelecanosEtAl2025}.  A single
run produces a classical outcome \(y\in\mathcal Y\), distributed according to
\(P(\cdot\mid C)\), together with an estimate \(\widehat C_y/d_I\) of the
input state.  The machine then prepares a fixed-gauge purification
\(\ket{\widehat V_y(C)}\) of \(\widehat C_y\), for instance by fixing the
environment gauge \(U_E=\id_E\).  Thus, conditioned on \(y\), the output is
\(\ketbra{\widehat V_y(C)}\).  If the classical outcome is forgotten, the
corresponding deterministic quantum operation outputs the averaged state
\begin{align}\label{eq::def_Qtomo}
    \Qcal_{\mathrm{tomo}}(C^{\otimes k})
    \coloneqq
    \int_{\mathcal Y}
    dP(y\mid C)\,
    \ketbra{\widehat V_y(C)} .
\end{align}
This channel is generally mixed and satisfies
\begin{align}
    \tr_E \left[\Qcal_{\mathrm{tomo}}(C^{\otimes k})\right]
    =
    \int_{\mathcal Y}
    dP(y\mid C)\,
   \widehat C_y.
\end{align}

Naturally, the error attained by such a strategy is
\begin{align}
\begin{split}\label{eq::def_estimation}
    &\epsilon_{\mathrm{tomo}} (d_I,d_O,d_E;k)\coloneqq \\
    &\EE_C \min_{U_E} \norm{\Qcal_{\mathrm{tomo}}(C^{\otimes k}) - \ketbra{V_{U_E}(C)}}^2_2.
\end{split}
\end{align}

By directly adapting fidelity-error and concentration bounds of Ref.~\cite{guta2018faststatetomographyoptimal,PelecanosEtAl2025} to the case of Choi operators of quantum channels, we show in App.~\ref{app::estimation} that

\begin{theorem}[Estimation-based strategy]\label{thm::estimation}
The average purification error $\epsilon_{\mathrm{tomo}}$ achieved by the estimation-based purification strategy $\Qcal_{\mathrm{tomo}}$ defined in Eq.~\eqref{eq::def_Qtomo} is upper bounded by
\begin{align}
\begin{split}
    &\epsilon_{\mathrm{tomo}}(d_I,d_O,d_E;k) \\
    &\leq
     4 d_I^2\min\left\{1,\kappa\,
    \frac{
        \min\{d_E,d_Id_O\}d_I d_O+\log(1/\delta)
    }{k}\right\}.
\end{split}
\end{align}
where $\kappa>0$ is a universal constant. Here $1-\delta$ denotes the success probability of the estimation strategy of Ref.~\cite{PelecanosEtAl2025}: with probability at least $1-\delta$, the output is an estimate $\hat C$ satisfying
\begin{align}\label{eq:epsilon-k}
    F(C/d_I,\hat C/d_I) \geq 1-\varepsilon_k,
\end{align}
with
\begin{align}
    \varepsilon_k
    :=
    \min\left\{1,\kappa\,
    \frac{
        \operatorname{rank}(C)d_I d_O+\log(1/\delta)
    }{k}\right\}.
\end{align}
\end{theorem}
\vspace*{11pt}

Note that the $\min\{d_E,d_Id_O\}=\rank(C)$, for randomly sampled $C$, depends only on the local dimensions and the prior $d_E$. As expected, in the limit of $k\to\infty$, $\epsilon_{\mathrm{tomo}}=0$.

\subsection{Comparison of all strategies}

To facilitate comparison, Table~\ref{tab::errors} collects the solution for the error attained by each of the analyzed strategies, Table~\ref{tab::d_Eregimes}  presents a comparison of the error of each strategy in the different regimes of $d_E=1$, $d_E=d_Id_O$, and $d_E\to\infty$, and finally  Table~\ref{tab::hierarchy} reports the hierarchy between these different strategies in the different regimes of $d_E$.

The optimal strategy depends strongly on the prior specified by $d_E$. For $d_E=1$, as expected, the append-environment strategy is optimal. Interestingly, and arguably counterintuitively, in this case we find that while non-optimal, a strategy that maps all isometric input channels to the fully depolarizing channel performs better than a strategy that maps all isometric input channels to a fixed pure output, even after optimizing over all fixed pure outputs. On the other hand, in the regime of $d_E\to\infty$, a pure-output strategy is optimal, while the append-environment strategy performs increasingly sub-optimally. As expected, the estimation-based strategy is optimal for all regimes of $d_E$ in the limit of $k\to\infty$, but always achieves a non-zero error for finite $k$, which we show that scales with $O(1/k)$.  

For the case of $d_I=d_O=2$, we explicitly compute the error values for the studied single-copy strategies and plot them in Fig.~\ref{fig::strategies} as a function of $d_E$. This allows one to see the explicit cross over between the error attained by append-environment and pure-output strategies, the former performing better for lower $d_E$ and the later performing better for higher $d_E$. We hence conclude that for small value of $d_E>1$, requiring that a universal purification machine necessarily prepares a pure output for every input quantum channel is not the optimal strategy.

\begin{table*}
  \centering
  {\renewcommand{\arraystretch}{2.5}
  \begin{tabular}{| c || c  c |}
    \hline
    Strategy & \multicolumn{2}{|c|}{Attainable error: $\epsilon(d_I,d_O,d_E;k)\leq$} \\
    \hline
    \hline
    Pure output &  $\epsilon_\mathrm{pure} = 2d_I^2 - \frac{2}{d_O}\EE_C\left\{\tr(\sqrt{C})^2\right\}$ &  [Sec.~\ref{sec::pureoutput}, App.~\ref{app::pureoutput}, Thm.~\ref{thm::pure-output}] \\
    \hline
    \multirow{1}{*}{Append environ.} & $\epsilon_\mathrm{app} = d_I^2-\frac{1}{\EE_C\{\tr(C^2)\}} \sum_{i}(\mathbb{E}_C\{(c_i^\downarrow)^2\})^2$  &  [Sec.~\ref{sec::appendenv}, App.~\ref{app::appendenvironment}, Thm.~\ref{thm::append-env}] \\
    \hline
    Map to depolar. & $\epsilon_\mathrm{dep}= d_I^2 - \frac{d_I}{d_Od_E}$ &  [Sec.~\ref{sec::maptodep}, App.~\ref{app::maptodepol}, Thm.~\ref{thm::maptodepol}] \\
    \hline
    $\,$ Gen. upper bound $\,$ & $\epsilon_{\mathrm{avg}-U_E} = d_I^2 - \frac{1}{d_E} \EE_C\{\tr(C^2)\}$ & [Sec.~\ref{sec::genupperbound}, App.~\ref{app::genuperbound}, Thm.~\ref{thm::genupperbound}] \\
    \hline
    Estimation based & $\,$ $\epsilon_{\mathrm{tomo}}\leq  4 d_I^2\min\left\{1,\kappa\,
    \frac{
        \min\{d_E,d_Id_O\}d_I d_O+\log(1/\delta)
    }{k}\right\}$& $\,$ [Sec.~\ref{sec::estimation}, App.~\ref{app::estimation}, Thm.~\ref{thm::estimation}] $\,$ \\
    \hline
    \hline
    \multicolumn{3}{|c|}{$\quad$ where $\EE_C\{\tr(C^2)\} = \frac{d_Id_O(d_E^2-1)+d_I^2d_E(d_O^2-1)}{d_O^2d_E^2-1}$, $\quad$ which satisfies $\frac{d_I}{d_O}\leq \EE_C\{\tr(C^2)\}\leq d_I^2$.} \\
    \hline
  \end{tabular}
  }
  \caption{{Summary of the attainable error of different strategies.}}
  \label{tab::errors}
\end{table*}

\begin{table*}
  \centering
  {\renewcommand{\arraystretch}{2.5}
  \begin{tabular}{| c || c | c | c | }
    \hline
    Strategy & $d_E=1$ & $d_E=d_Id_O$ & $d_E\to\infty$ \\
    \hline
    \hline
    Pure output & $\quad$ $2\left(d_I^2 - \frac{d_I}{d_O}\right)$ $\quad$ & 
    [see App.~\ref{ap:: reviewrandomchannels}, Lemma~\ref{lemma:: asymptotics trace sqrt of C}]  & 0 \\
    \hline
    Append environ. & 0 & $\quad$ $d_I^2 - E(d_I,d_O) \leq \epsilon_\mathrm{app} \leq d_I^2 - \frac{1}{d_Id_O} E(d_I,d_O) $ $\quad$ & $\quad$ $d_I^2 - \frac{1}{d_O^2}$ $\quad$ \\
    \hline
    Map to depolar. & $d_I^2-\frac{d_I}{d_O}$ & $d_I^2 - \frac{1}{d_O^2}$ & $d_I^2$ \\
    \hline
    $\quad$ Gen. upp. bound $\quad$ & 0 & $d_I^2 - \frac{1}{d_Id_O}E(d_I,d_O)$ & $d_I^2$ \\
    \hline
    \hline
    \multicolumn{4}{|c|}{where $E(d_I,d_O) \coloneqq \frac{d_Id_O(d_I^2d_O^2-1)+d_I^3d_O(d_O^2-1)}{d_O^4d_I^2-1}$.} \\
    \hline
  \end{tabular}
  }
  \caption{{Attainable error of different strategies in different regimes of $d_E$.}}
  \label{tab::d_Eregimes}
\end{table*}

\begin{table*}
  \centering
  {\renewcommand{\arraystretch}{2.5}
  \begin{tabular}{| c || c  c |}
    \hline
    Regime & \multicolumn{2}{c|}{Hierarchy} \\
    \hline
    \hline
    $d_E=1$ & $\quad$ $0=\epsilon_\mathrm{app}=\epsilon_{\mathrm{avg}-U_E}\leq\epsilon_\mathrm{dep}\leq\epsilon_\mathrm{pure}$ & $\quad$ $\forall \ d_I,d_O$ $\quad$ \\
    \hline
     \multirow{2}{*}{$\quad$ $d_E=d_Id_O$} $\quad$ & $\quad$ $\epsilon_\mathrm{app}\leq\epsilon_{\mathrm{avg}-U_E}\leq\epsilon_\mathrm{dep}$ & $\quad$ $\forall \ d_I,d_O$ $\quad$  \\
    & $\quad$ $0<\epsilon_\mathrm{pure}<\epsilon_\mathrm{app}<\epsilon_{\mathrm{avg}-U_E}<\epsilon_\mathrm{dep}$ & $\ \ \ \ d_I=d_O=2$ $\quad$ \\
    \hline
    $d_E\to\infty$ & $\quad$ $0=\epsilon_\mathrm{pure}\leq\epsilon_\mathrm{app}\leq\epsilon_\mathrm{dep}=\epsilon_{\mathrm{avg}-U_E}$ & $\quad$ $\forall \ d_I,d_O$ $\quad$ \\
    \hline
  \end{tabular}
  }
  \caption{{Hierarchy of the performance of different strategies.}}
  \label{tab::hierarchy}
\end{table*}

\begin{figure}[t]
    \begin{center}
    \includegraphics[width=\columnwidth]{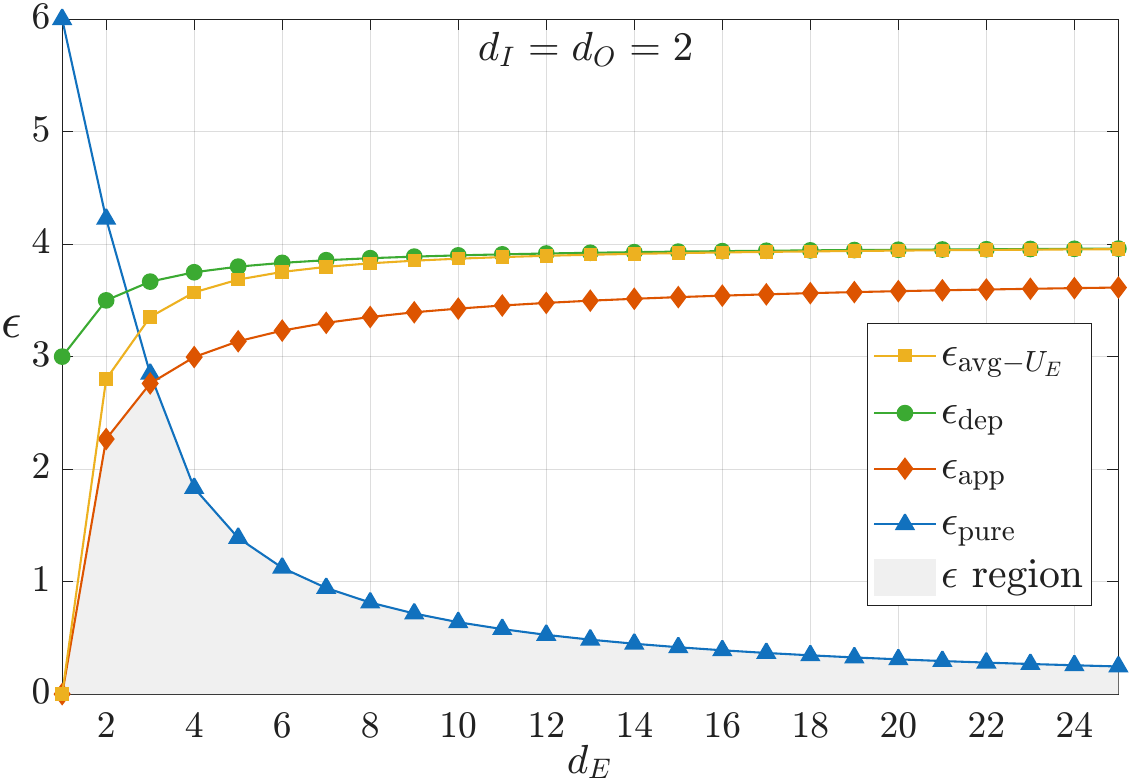}
    \caption{\textbf{Comparison of the error attained by different strategies in the case of qubit-qubit channels.} Plot of the analytical values of the average error attained in the case of $d_I=d_O=2$, and $d_E\in\{1,\ldots,25\}$ by the following single-copy strategies: the pure-output strategy $\epsilon_{\mathrm{pure}}$ in Eq.~\eqref{eq::epsilon_pureoutput_solution}, the append-environment strategy $\epsilon_{\mathrm{app}}$ in Eq.~\eqref{eq::epsilon_app_solution}, the map-to-depolarizing strategy $\epsilon_{\mathrm{dep}}$ in Eq.~\eqref{eq::epsilon_maptodep_solution}, and the general error upper bound $\epsilon_{\mathrm{avg-}U_E}$ in Eq.~\eqref{eq::epsilon_avgUE_solution} (see Table~\ref{tab::errors} for summary). All curves constitute upper bounds for the minimum average error, implying that $\epsilon(2,2,d_E;k)$ must be in the gray region of the plot.} 
\label{fig::strategies}
\end{center}
\end{figure}

\subsection{Hardness of computing general error lower bounds} 
In full generality, we do not expect the quantity
\begin{align} 
\begin{split} 
    & g(\Qcal,V) \coloneqq \\ 
    & \max_{U_E}\left\{\tr(\Qcal\left(\tr_E(\ketbra{V})^{\otimes k}\right)(\id\otimes U_E)\ketbra{V}(\id\otimes U_E^*))\right\}, 
\end{split} 
\end{align}
which corresponds to the overlap term in Eq.~\eqref{eq::error_expanded}, to admit a closed-form analytical expression. The main difficulty is that the maximization ranges over the full local-unitary orbit of the purification $V(C)$, leading to a nonconvex optimization problem whose value depends on the detailed operator structure of $\Qcal[\tr_E(\ketbra{V})^{\otimes k}]$, and not only on coarse spectral data. In particular, once $\ket{V}$ is written in a Schmidt basis, the above expression takes the form of a quadratic optimization problem in the matrix elements of $U_E$, subject to the nonlinear unitary constraints $U_E^* U_E=\id$. Without additional symmetry, covariance, or algebraic structure of $\Qcal$, there is therefore no canonical reduction of the problem that would lead to an exact analytic formula. For this reason, attempting to compute $g(\Qcal,V)$ exactly is, in general, non-tractable. Attempts to lower bound this quantity with standard techniques wind up decoupling the term $\Qcal[\tr_E(\ketbra{V})^{\otimes k}]$ from the term $(\id\otimes U_E)\ketbra{V}(\id\otimes U_E^*)$, such as e.g. using the Cauchy-Schwarz inequality, and leading to trivial lower bounds of $\epsilon\geq 0$. It has instead been more tractable to derive explicit upper bounds, attained by the physically well-motivated strategies detailed above, which are analytically accessible and provide the error estimates presented here.

\section{Comparison with cloning, controlization, and average purification}

The impossibility of universal purification is closely related in spirit to the celebrated no-cloning theorem~\cite{wootters1982asingle}. Both results establish fundamental restrictions imposed by the linear structure of quantum theory. The no-cloning theorem states that there exists no deterministic device that can take as input an arbitrary unknown pure state and output two perfect copies of it. Similarly, there exists no deterministic device that can take as input an arbitrary mixed state or noisy channel and output one of its purifications. In both cases, universality is not essential for the contradiction: no-cloning already follows from two non-orthogonal pure states, and our no-purification result already follows from a rank-one and a rank-two state.

Another similar no-go theorem in quantum theory is the impossibility of controlization of unknown unitaries~\cite{araujo2014quantum,gavorova2024topological}. However, exact controlization becomes possible when partial information about the unitary is available, for instance, if one of its eigenstates or its action on a reference state is known~\cite{araujo2014quantum,gavorova2024topological}.
This partial knowledge is responsible for the implementation of universal controlization of unitary channels in optical interferometers~\cite{friis2014implementing} or using quantum circuits~\cite{dong2019controlled}. In the case of purification, on the other hand, we expect that no amount of partial information would be sufficient to reconstruct all the information that is lost in a partial trace, allowing for an exact and deterministic purification. This argument points toward a potentially deeper obstruction of purification of quantum states and channels than that of controlization of unitaries.\\

In the present work, we consider the transformation of many copies of an arbitrary input into a single copy of any of its possible purifications. This can be viewed as a purification task in which, given many copies of an input, one aims to randomly sample a single copy of one of its possible purifications. This problem stands sharp contrast to a related task studied in Refs.~\cite{chen2024localtestunitarilyinvariant,chen2025quantumchanneltomographyestimation,tang2025,girardi2025random,girardi2025random2,yoshida2025}. There, the goal is to transform many copies of an input into the average over many copies of its possible purifications, rather than to sample a single purification. 

Unlike our task, the production of such an average is compatible with the linearity and positivity of quantum theory. For this reason, it is physically realizable and can be implemented efficiently with quantum circuits~\cite{chen2024localtestunitarilyinvariant,chen2025quantumchanneltomographyestimation,tang2025,girardi2025random,girardi2025random2,yoshida2025}. These protocols have applications to the optimal tomography of mixed states and non-unitary channels. In fact, it is known for property-testing problems in general, having access to a purification of a quantum state or channel does not offer advantage compared to its mixed form~\cite{chen2024localtestunitarilyinvariant,chen2025quantumchanneltomographyestimation}. However, when the average-purification protocol of Refs.~\cite{chen2024localtestunitarilyinvariant,chen2025quantumchanneltomographyestimation,tang2025,girardi2025random,girardi2025random2,yoshida2025} is applied to our setting, it is effectively equivalent to appending a maximally mixed state to the environment, as discussed in Sec.~\ref{app::appendenvironment}. As we show here, this strategy does not, in general, provide a good approximation to a randomly sampled purification of the input channel. 

This suggests the existence of alternative quantum information tasks, unlike property testing~\cite{chen2024localtestunitarilyinvariant,chen2025quantumchanneltomographyestimation}, in which access to a randomly sampled purification---and to an optimal universal approximate purification machine---could outperform the average purification protocol of Refs.~\cite{tang2025,girardi2025random,girardi2025random2,yoshida2025}. Characterizing such tasks remains an interesting direction for future work.

\section{Conclusions}

In this work, we investigate the problem of purifying arbitrary quantum states and channels, with channel purification understood in the sense of the Stinespring dilation, both in the probabilistic exact and deterministic approximate settings, establishing a framework of quantum purification machines. In the probabilistic exact setting, we showed that no machine that has access to a finite number of copies of the input can purify even two fixed inputs with non-zero probability while remaining a linear positive map. As a corollary, we recover that universal probabilistic purification of any finite number of input copies is impossible with a non-zero probability.

Turning to the approximate setting, we formalize approximate quantum purification machines as higher-order transformations acting on black-box quantum states and channels. The error attained by these machines is specified by the input and output dimensions, the environment dimension of the target purification, and the number of available copies of the unknown input. Importantly, we allowed the output to be an arbitrary approximation to a valid purification, without requiring it to be pure or isometric. 

Using the minimal squared Hilbert–Schmidt distance over all possible purifications as our figure of merit. We defined an average purification error by averaging over Haar–Stinespring distributed quantum channels and states, thereby interpreting the environment dimension as a prior distribution over inputs. We derived analytical expressions for physically motivated strategies, namely pure-output strategies, append-environment strategies, a map-to-depolarizing strategy, and an estimation-based many-copy strategy, along with a general upper bound on the error, which is tight in certain regimes. 

Comparing these strategies, we found that in the regime of large environment dimension, strategies that output a fixed isometric channel or pure state perform best among those considered, whereas for small environment dimensions, strategies that append a state to the environment---typically producing a non-pure output---achieve better performance. Within the class of deterministic pure-output strategies, we proved that the optimal one outputs one of the possible purifications of the fully depolarizing channel. We furthermore showed how the performance of such machines can be understood as a function of the amount of entanglement between the input/output systems and the environment in the fixed pure output. 
 
In the many-copy setting, we analyzed an estimation-based strategy and showed that it attains a purification error that scales with the inverse of the number of copies. However, it remains an open question whether the globally optimal strategy in the many-copy scenario is indeed estimation-based, or whether a more intrinsically quantum approach, which does not primarily aim to extract a classical description of the input, could achieve better performance in this purification task.

Another open direction is to extend this framework beyond states and channels to measurements and quantum instruments. It would be interesting to determine how well universal purification of such quantum operations can be approximated in deterministic or probabilistic settings.

\vspace*{11pt}
\noindent\textit{Acknowledgments.}
We are grateful to Ion Nechita and Marco T\'{u}lio Quintino for helpful discussions. We acknowledge financial support from the
French national quantum initiative managed by
Agence Nationale de la Recherche (ANR) in the frame-work of France 2030 through project EPIQ, with reference ANR-22-PETQ-0007.


%

\clearpage 
\appendix
\onecolumngrid

\section*{APPENDIX}
\setcounter{theorem}{1}

\section{Random Choi operators}\label{ap:: reviewrandomchannels}

We collect here the random-matrix facts used throughout the paper. Standard references are
\cite{Collins_Nechita_2016,kukulski2021generating,MarchenkoPastur1967,CollinsNechita2011}.

We begin with the ensemble of quantum channels induced by Haar-distributed Stinespring isometries, what we will refer to the Haar-Stinespring ensemble in what comes for the sake of simplicity. A Haar-random isometry is an isometry
\begin{align}
    V:\Hcal_{I'}\to \Hcal_O\otimes \Hcal_E
\end{align}
distributed according to the Haar measure on the complex Stiefel manifold
\begin{align}
    \mathrm{St}(d_I,d_Od_E)\coloneqq\{V\in M_{d_Od_E\times d_I}(\mathbb C):V^\ast V=\id_{I'}\}.
\end{align}
Equivalently, if a unitary $U\in \mathcal U(d_Od_E)$ is Haar-distributed, then $V$ may be taken as the first $d_I$ columns of $U$. The associated Choi vector is given by
\begin{align}
    \ket{V}_{IOE}\coloneqq(\id_I\otimes V)\ket{\Phi^+}\in \Hcal_I\otimes \Hcal_O\otimes \Hcal_E,
\end{align}
where $\ket{\Phi^+}=\sum_i\ket{ii}$ is the unnormalized maximally entangled state.\\

Furthermore, this isometry represents the Stinespring dilation of a random channel, hence
\begin{align}
    C_V\coloneqq\tr_E(\ketbra{V})\in \Lcal(\Hcal_I\otimes \Hcal_O).
\end{align}
By construction, $C_V$ is the Choi operator of the channel
\begin{align}
    \widetilde{C}_V(X)\coloneqq\tr_E(VXV^\ast),
\end{align}
and therefore
\begin{align}
    C_V\ge 0,\qquad \tr_O(C_V)=\id_I,\qquad \tr(C_V)=d_I.
\end{align}
In particular,
\begin{align}
    \mathbb E_V\{\,\tr(C_V)\}=d_I.
\end{align}
Moreover, by unitary invariance and the constraint $\tr_O(C_V)=\id_I$, one also has
\begin{align}
    \mathbb E_V\{C_V\}=\frac{\id_{IO}}{d_O}.
\end{align}
From this point of view, low-order moments of Haar-random Choi operators can be computed using either Weingarten calculus on the Stiefel side; see Refs.~\cite{CollinsSniady2006,Collins_Nechita_2016}

\begin{lemma}\label{lemma:: average purity of C}
    Let $C\in\Lcal(\Hcal_{I}\otimes\Hcal_{O})$ be a Choi operator sampled from the so-called Haar-Stinespring uniform measure. Then
    \begin{align}\label{eq::avgpurity_app}
        \EE_C\{\tr(C^2)\}&=\frac{d_Id_E(d_O^2-1)+d_I^2d_O(d_E^2-1)}{d_O^2d_E^2-1}.
    \end{align}
\end{lemma}

\begin{proof}
This is a consequence of the findings in random quantum channel theory and Haar-random isometries found in Refs.~\cite{Collins_Nechita_2016,kukulski2021generating,Puchala_Miszczak_2017}.\\

First step is to linearize the expression $\tr(C^2)$. Define the swap operator (also called flip) as the unitary $F_{A,B}: \mathcal{H}_A \otimes \mathcal{H}_B \to \mathcal{H}_A \otimes \mathcal{H}_B $ defined on pure tensors by
\begin{equation}
    F_{A,B}\left(\ket{a}_A\otimes\ket{b}_B\right) \coloneqq \ket{b}_A\otimes\ket{a}_B,
\end{equation}
for all $\ket{a}\in \mathcal{H}_A$ and $\ket{b}\in \mathcal{H}_B$. Equivalently, in the chosen bases,  
\begin{equation}
   F_{A,B}=\sum_{i, j=1}^d\ket{i}_A\bra{j}_A \otimes \ket{j}_B\bra{i}_B=\sum_{i, j=1}^d\ketbra{ij}{ji}_{AB},
\end{equation}
so that its matrix elements satisfy $\left(F_{A,B}\right)_{i j, k l}=\delta_{i l} \delta_{j k}$.

The swap trick then consists in the following identity: for any $X,Y\in \Lcal(\mathcal{H})$, we have that
\begin{equation}
    \operatorname{Tr}(XY)= \operatorname{Tr}[(X \otimes Y) F].
\end{equation}

Note that the trick requires the two tensor factors to have the same dimension (so that $F$ is a unitary on $\Hcal\otimes \Hcal$).

Then, 
\begin{align}
     \mathbb{E}_C \big\{\tr\left(C^2\right)\big\} = \tr\Big[F_{IO,I'O'} \left(\mathbb{E}_C\big\{C_{IO}\otimes C_{I'O'}\big\}\right) \Big]. \label{eq::g_post_swap}
\end{align}

We are now ready to formally define the average function $\mathbb{E}_C$, in order to compute $\mathbb{E}_C\{C_{IO}\otimes C_{I'O'}\}$.

Rather than averaging directly over Choi operators, we pass through their purifications. Concretely, we draw a Haar-random isometry that implements a Stinespring dilation, build the corresponding purified Choi state, and take the partial trace over the environment. Then,
\begin{equation}
    C_{IO}=\tr_{E}\left[\left(\id_{IO}\otimes U_{E}(C)\right)\ketbra{V(C)}_{IOE}\left(\id_{IO}\otimes U^*_{E}(C)\right)\right]=\tr_{E}\left[\ketbra{V(C)}_{IOE}\right],
\end{equation}
which is independent of the local unitary $U_{E}$ since the environment space is traced out.

The average over quantum channels $C$ is then defined via the average over the isometries $V$ 
\begin{align}
    \mathbb{E}_C\{C_{IO}\otimes C_{I'O'}\} 
    &= \mathbb{E}_{V}\{\tr_{E}(\ketbra{V}_{IOE})\otimes \tr_{E'}(\ketbra{V}_{I'O'E'})\} \\
    &= \mathbb{E}_{V}\{\tr_{{E}{E}'} \left(\ketbra{V}_{IOE}\otimes\ketbra{V}_{I'O'E'}\right)\} \\
    &= \tr_{{E}{E}'}\left(\mathbb{E}_{V}\left\{\ketbra{V}_{IOE}\otimes\ketbra{V}_{I'O'E'}\right\} \right) \label{eq::partialtrace_EE}
\end{align}

Now we are only left to compute the expected value $\mathbb{E}_{V_C}\left\{\ketbra{V_C}_{IOE}\otimes\ketbra{V_C}_{I'O'\tilde E'}\right\} = \mathbb{E}_{V}\{\ketbra{V}^{\otimes 2}\}$ over Haar random isometries. Take $D\coloneq d_Od_{E}$ to be the dimension of the total output space of the isometries. The required second-moment identity is standard in the literature of Haar integration and Weingarten calculus~\cite{CollinsSniady2006,Collins_Nechita_2016,Puchala_Miszczak_2017}, and given by
\begin{align}\label{eq:2moment}
    \mathbb{E}_{V}\{\ketbra{V}^{\otimes 2}\}
    &=
    \int_{\mathrm{St}(d_I,D)} \ketbra{V}_{IOE} \otimes \ketbra{V}_{I'O'E'}\, dV \\
    &= 
    2\left( \frac{ P^{+}_{I,I'} \otimes P^{+}_{OE,O'E'}}{D^2+D} \right) + 2\left( \frac{P^{-}_{I,I'} \otimes P^{-}_{OE,O'E'}}{D^2-D}  \right) \\
    &=
    \frac{1}{D^2-1} (\id_{II'OEO'E'} + F_{I,I'}\otimes F_{OE,O'E'})- \frac{1}{D(D^2-1)} (F_{I,I'}\otimes\id_{OE,O'E'} + \id_{II'}\otimes F_{OE,O'E'}),
\end{align}
where $P^{\pm}_{A,B}=\frac{\id_{AB}\pm F_{A,B}}{2}$ are the projectors on the symmetric and antisymmetric subspaces. We then reintroduce this expression in Eq.~\eqref{eq::partialtrace_EE} to perform the partial traces over $E$ and $E'$, arriving at 
\begin{align}
    \mathbb{E}_C\{C_{IO}\otimes C_{I'O'}\} &= \tr_{{E}{E}'} \left[ \mathbb{E}_{V}\{\ketbra{V}^{\otimes 2}\} \right] \\
    &= \frac{1}{(D^2-1)}\left(d_{E}^2\,\id_{II'OO'} + d_{E}\, F_{II'}\otimes F_{OO'} \right) - \frac{1}{D(D^2-1)}\left(d_{E}^2 \,F_{II'}\otimes\id_{OO'} + d_{E}\,\id_{II'}\otimes F_{OO'} \right).
\end{align}
Plugging it into \eqref{eq::g_post_swap},
\begin{align}
    \EE_C & \tr[F_{II',OO'}\left(\frac{1}{(D^2-1)}\left(d_{E}^2\,\id_{II'OO'} + d_{E}\, F_{I,I'}\otimes F_{O,O'} \right) - \frac{1}{D(D^2-1)}\left(d_{E}^2 \,F_{I,I'}\otimes\id_{OO'} + d_{E}\,\id_{II'}\otimes F_{O,O'} \right)\right)]\nonumber \\
    &=\frac{1}{d_E^2d_O^2-1}\left(d_E^2d_Id_O+d_Ed_I^2d_O^2-\frac{d_E^2d_I^2d_O+d_Ed_Id_O^2}{d_Ed_O}\right)\\
    &=\frac{d_Id_O(d_E^2-1)+d_I^2d_E(d_O^2-1)}{d_O^2d_E^2-1}.
\end{align}
\end{proof}

The Haar-Stinespring model admits a convenient Gaussian realization. Let $G\in M_{m\times n}(\mathbb C)$ be a standard complex Ginibre matrix, meaning that its entries are i.i.d. complex Gaussian random variables with law $\mathcal N_{\mathbb C}(0,1)$, equivalently
\begin{align}
    \Re (G_{ij}),\Im (G_{ij})\sim \mathcal N(0,1/2)
    \quad\text{independently}.
\end{align}
If
\begin{align}
    G\in M_{d_Od_E\times d_I}(\mathbb C)
\end{align}
is Ginibre and $d_Od_E\ge d_I$, then $G^\ast G$ is invertible almost surely, and the polar factor
\begin{align}
    V\coloneqq G(G^\ast G)^{-1/2}
\end{align}
is Haar-distributed on $\mathrm{St}(d_I,d_Od_E)$. Thus Haar-random Stinespring isometries can be sampled from Ginibre matrices.

This Gaussian representation is especially useful because it makes the Kraus structure completely explicit. Writing $G$ in environment blocks as
\begin{align}
    G=\sum_{\alpha=1}^{d_E} G_\alpha\otimes \ket{\alpha}_E,
    \qquad
    G_\alpha\in M_{d_O\times d_I}(\mathbb C),
\end{align}
and defining
\begin{align}
    S\coloneqq G^\ast G=\sum_{\alpha=1}^{d_E} G_\alpha^\ast G_\alpha,
    \qquad
    A_\alpha\coloneqq G_\alpha S^{-1/2},
\end{align}
one obtains the Kraus representation
\begin{align}
    \widetilde{C}_V(X)=\sum_{\alpha=1}^{d_E} A_\alpha X A_\alpha^\ast.
\end{align}
In particular, the Stinespring, Kraus, and Choi descriptions all refer to the same random-channel ensemble.\\

A second useful viewpoint is obtained by expressing the random Choi operator itself through a normalized Wishart construction. Let $ \mathsf{C}\in \Lcal(\Hcal_I\otimes\Hcal_O)$ be a Wishart matrix with parameters $(d_Id_O,d_E)$, that is,
\begin{align}
    \mathsf{C}=GG^\ast
\end{align}
for a Ginibre matrix $G\in M_{d_Id_O\times d_E}(\mathbb C)$. Define
\begin{align}
    T\coloneqq \tr_O( \mathsf{C})\in \Lcal(\Hcal_I),
    \qquad
    C\coloneqq (T^{-1/2}\otimes \id_O)\,\mathsf{C}\, (T^{-1/2}\otimes \id_O).
\end{align}
Then $C$ satisfies $\tr_O(C)=\id_I$, and its distribution coincides with the distribution of $C_V$ arising from a Haar-random Stinespring isometry. This equivalence provides the basic bridge between random-channel theory and Wishart ensembles.\\

We will use the Gaussian representation to define induced states and normalized Wishart matrices, which govern the spectral behavior of random reduced states and random channel outputs. Let $n,s\in\mathbb N$, let
\begin{align}
    \ket{\psi}\in \mathbb C^n\otimes \mathbb C^s
\end{align}
be Haar-distributed on the unit sphere, and define the reduced state
\begin{align}
    \rho\coloneqq\tr_{\mathbb C^s}\ketbra{\psi}\in M_n(\mathbb C).
\end{align}
Then $\rho$ has the induced measure of parameters $(n,s)$, and its distribution is equivalent to that of a normalized Wishart matrix,
\begin{align}
    \rho \overset{d}{=} \frac{W}{\tr W},
    \qquad
    W=GG^\ast,
    \qquad
    G\in M_{n\times s}(\mathbb C)\ \text{Ginibre}.
\end{align}
In particular, when $s\ge n$, the joint density of the unordered eigenvalues $\lambda_1,\dots,\lambda_n\ge 0$ of $\rho$ is proportional to
\begin{equation}
    \delta\!\left(1-\sum_{i=1}^n\lambda_i\right)
    \prod_{i=1}^n \lambda_i^{\,s-n}
    \prod_{1\le i<j\le n}(\lambda_i-\lambda_j)^2.
\label{eq:induced-density}
\end{equation}
This ensemble appears repeatedly in quantum information theory because random reduced states, as well as the outputs of Haar-random channels acting on fixed pure inputs, are distributed as normalized Wishart matrices.

In the random-channel setting this connection takes a particularly simple form. If $\Phi_V:\Lcal(\Hcal_I)\to\Lcal(\Hcal_O)$ is Haar-Stinespring random and $\ketbra{x}$ is a fixed pure input state, then the output
\begin{align}
    \Phi_V(\ketbra{x})
\end{align}
is distributed as an induced state of parameters $(d_O,d_E)$, equivalently as
\begin{align}
    \frac{GG^\ast}{\tr(GG^\ast)},
    \qquad
    G\in M_{d_O\times d_E}(\mathbb C)\ \text{Ginibre}.
\end{align}
Thus one can transfer spectral information from normalized Wishart matrices directly to random channel outputs.

Finally, for asymptotic spectral estimates we use the Mar\v{c}enko--Pastur law. Let $G_n\in M_{n\times s_n}(\mathbb C)$ be Ginibre and define the rescaled Wishart matrices
\begin{align}
    W_n\coloneqq \frac1{s_n}G_nG_n^\ast.
\end{align}
Assume that
\begin{align}
    \frac{n}{s_n}\longrightarrow c\in(0,\infty)
    \qquad\text{as }n\to\infty.
\end{align}
Then the empirical eigenvalue distribution
\begin{align}
    \mu_{W_n}\coloneqq \frac1n\sum_{i=1}^n\delta_{\lambda_i(W_n)}
\end{align}
converges almost surely and weakly to the Mar\v{c}enko--Pastur law $\mu_{\mathrm{MP},c}$ \cite{MarchenkoPastur1967}, given by
\begin{equation}
    \mu_{\mathrm{MP},c}
    =
    \left(1-\frac1c\right)_+\delta_0
    +
    \frac{\sqrt{(x_+-x)(x-x_-)}}{2\pi c\,x}\,
    \mathbf 1_{[x_-,x_+]}(x)\,dx,
    \qquad
    x_\pm=(1\pm \sqrt c)^2.
\label{eq:MP-law}
\end{equation}
Since
\begin{align}
    \frac{\tr(G_nG_n^\ast)}{ns_n}\xrightarrow[n\to\infty]{a.s.}1,
\end{align}
the same scaling applies to normalized Wishart matrices
\begin{align}
    \rho_n\coloneqq \frac{G_nG_n^\ast}{\tr(G_nG_n^\ast)}.
\end{align}
Indeed,
\begin{align}
    n\rho_n=\frac{ns_n}{\tr(G_nG_n^\ast)}\,W_n,
\end{align}
and the prefactor converges almost surely to $1$. Consequently, the empirical eigenvalue distribution of $n\rho_n$ also converges almost surely to $\mu_{\mathrm{MP},c}$. Equivalently, the eigenvalues of $\rho_n$ are typically of order $1/n$, and after multiplication by $n$ they fill the Mar\v{c}enko--Pastur bulk.

\begin{lemma}\label{lemma:: asymptotics trace sqrt of C}
     Let $C\in\Lcal(\Hcal_{I}\otimes\Hcal_{O})$ be a Choi operator sampled from the uniform measure induced by Haar-distributed Stinespring isometries , then
    \begin{align}
        \EE_C\left\{\tr(\sqrt{C})^2\right\}&=O(d_I^2d_O)
    \end{align}
    in the asymptotic regime for $d_E\geq d_Id_O$
\end{lemma}

\begin{proof}
Using normalized Wishart matrices, we construct the same Choi distribution using $d_E$ i.i.d. Ginibre Kraus operators followed by trace-preserving normalization, i.e.\ a partially normalized Wishart Choi matrix \cite[Thm.~1, Eq.~(4)]{KNPPZ21}.\\

Concretely, one may sample a complex Ginibre matrix $G\in\mathbb{C}^{d_Id_O\times d_E}$ (i.i.d.\ standard complex Gaussians),
form the Wishart matrix $W \coloneqq \frac{1}{d_E}GG^*\in \Lcal(I\otimes O)$, and then enforce
$\mathrm{Tr}_O[C]=\id_I$ by the partial normalization
\begin{align}
    C \;\equiv\; \frac{1}{d_O}\Big( \big(\mathrm{Tr}_O W\big)^{-1/2}\otimes \id_O\Big)\,W\,
    \Big( \big(\mathrm{Tr}_O W\big)^{-1/2}\otimes \id_O\Big),
\end{align}
which matches Haar–Stinespring sampling (as an induced measure), for integer environment size $d_E$.\\

Assume a proportional-growth regime in which $d_Id_O\to\infty$ and
\begin{align}
    c \;\coloneqq \; \frac{d_Id_O}{d_E} \;\to\; c_0\in(0,\infty),
\end{align}
(and $d_I,d_O,d_E\to\infty$ accordingly).

Then the operator-norm effect of the partial normalization is asymptotically negligible \cite[Prop.~11]{NPPZ18}:
\begin{align}
    \big\|\, d_O C - W \,\big\|_\infty \xrightarrow[]{a.s.} 0,
\end{align}
so $d_O C$ has the same limiting empirical eigenvalue distribution as $W$.
By the Mar\v{c}enko--Pastur theorem, the empirical eigenvalue distribution of $W$ converges a.s.\ to $\mathrm{MP}_{c_0}$
with density~\cite{AGZ10}
\begin{align}
    \mathrm{MP}_{c}(dx)
    =
    \Big(1-\frac{1}{c}\Big)_+ \delta_0(dx)
    +
    \frac{\sqrt{(b-x)(x-a)}}{2\pi c\,x}\,\mathbf{1}_{[a,b]}(x)\,dx,
    \qquad
    a=(1-\sqrt{c})^2,\; b=(1+\sqrt{c})^2.
\end{align}

Since $\sqrt{\cdot}$ is continuous on $[0,b]$ and the spectrum of $d_O C$ is asymptotically supported in $[0,b]$,
we obtain the almost sure limit of the linear statistic
\begin{align}
    \frac{1}{d_Id_O}\mathrm{Tr}\sqrt{d_O C}
    \;\xrightarrow[]{a.s.}\;
    \mu(c_0),
    \qquad
    \mu(c)\coloneqq\int \sqrt{x}\,\mathrm{MP}_{c}(dx)
    = \int_a^b \frac{\sqrt{(b-x)(x-a)}}{2\pi c\,\sqrt{x}}\,dx.
\end{align}
Consequently, writing $D:=d_I d_O$,
\begin{align}
    \frac{1}{d_I^2 d_O}\,\big(\mathrm{Tr}\sqrt{C}\big)^2
    =
    \Big(\frac{1}{D}\mathrm{Tr}\sqrt{d_O C}\Big)^2
    \;\xrightarrow[]{a.s.}\;
    \mu(c_0)^2,
\end{align}
and by the deterministic bound $\frac{1}{d_I^2 d_O}(\mathrm{Tr}\sqrt{C})^2\le 1$,
dominated convergence yields convergence of expectations:
\begin{align}
    \mathbb{E}\big(\mathrm{Tr}\sqrt{C}\big)^2
    =
    d_I^2 d_O\,\mu(c_0)^2 \;+\; o\!\left(d_I^2 d_O\right).
\end{align}
Then, we are able to get a closed form for $\mu(c)$ in terms of elliptic integrals \cite{KNPPZ21}. Let  $m\coloneqq\frac{4r}{(1+\sqrt c)^2}\in(0,1]$.
With the standard complete elliptic integrals $K(m),E(m)$,
\begin{align}
    \mu(c)=\frac{2(1+\sqrt{c})}{3\pi c}\Big((1+c)E(m)-(1-\sqrt{c})^2K(m)\Big).
\end{align}
In particular, at the balanced point $c=1$ (i.e.\ $d_E=d_Id_O$),
\begin{align}
    \mu(1)=\frac{8}{3\pi},
    \qquad
    \mathbb{E}\big(\mathrm{Tr}\sqrt{C}\big)^2
    =
    d_I^2 d_O \cdot \frac{64}{9\pi^2}
    \;+\; o(d_I^2 d_O).
\end{align}

In the large environment regime [$d_E\ge d_Id_O$ (i.e.\ $c\le 1$)], using the small-$c$ expansion $\mu(c)=1-\frac{c}{8}+O(c^2)$,
\begin{align}
    \mathbb{E}\big(\mathrm{Tr}\sqrt{C}\big)^2
    =
    d_I^2 d_O\Big(1-\frac{1}{4}\frac{D}{d_E}+O\!\big((D/d_E)^2\big)\Big)
    =
    d_I^2 d_O - \frac{d_I^3 d_O^2}{4d_E} + O\!\Big(\frac{1}{d_E^2}\Big).
\end{align}
Then, as $d_E\to\infty$ (with fixed $d_I,d_O$), the limit is $d_I^2 d_O$, consistent with
concentration of $C$ near $\id_{IO}/d_O$.

\end{proof}

To summarize, the Haar-Stinespring model is the intrinsic one for random channels and pure Choi operators, while the Ginibre/Wishart model is the most convenient one for explicit calculations. Exact finite-dimensional quantities such as $\EE\,\tr(C_V)$ and $\EE\,\tr(C_V^2)$ are naturally handled by Wick or Weingarten calculus, whereas asymptotic spectral information is described by Wishart theory and the Mar\v{c}enko--Pastur law.

\section{Pure-output strategy}\label{app::pureoutput}

\begin{figure}[H]
    \centering
    \includegraphics[width=0.4\linewidth]{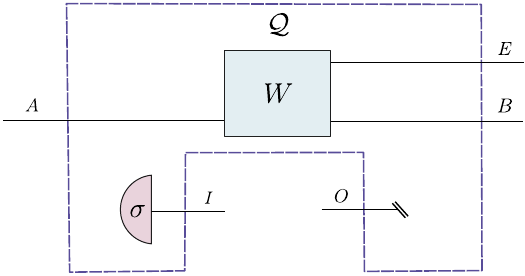}
    \caption{\textbf{Pure-output strategies.} Depiction of the family of strategies that forgets the input channel and outputs a fixed isometric channel $\widetilde{W}$.}
    \label{fig:map to pure W}
\end{figure}

Here we analyze the average purification error attained by the single-copy pure-output strategies, which output a fixed isometric channel regardless of the input quantum channel. We recall that these strategies, defined in Sec.~\ref{sec::pureoutput}, are described by a linear map $\Qcal_{\ket{W}}:\Lcal (\Hcal_I\otimes\Hcal_O)\rightarrow \Lcal (\Hcal_A \otimes\Hcal_B \otimes \Hcal_E)$ such that
\begin{equation}
        \Qcal_{\ket{W}}(C)=\tr(C)\frac{\ketbra{W}}{d_I} = \ketbra{W}
\end{equation}
and $\ketbra{W}\in\Lcal(\Hcal_{A}\otimes\Hcal_{B}\otimes\Hcal_{E})$ is a fixed isometric channel and $\Hcal_{A}\cong\Hcal_{I}$, $\Hcal_{B}\cong\Hcal_{O}$, implying $d_A=d_I$ and $d_B=d_O$. Such strategies attain an error defined by
\begin{align}
    \epsilon_{\ket{W}}(d_I,d_O,d_E) &= \EE_C \min_{U_E} \norm{\Qcal_{\ket{W}}(C)-\ketbra{V_{U_E}(C)}}^2_2 \\
    & = \EE_C \min_{U_E}\norm{\ketbra{W}-\ketbra{V_{U_E}(C)}}^2_2.
\end{align}

We are interested in finding out the minimum error of pure-output strategies, taken over all fixed output isometries $\ket{W}$
\begin{align}
    \epsilon_\mathrm{pure}(d_I,d_O,d_E) &\coloneqq \min_{\ket{W}} \epsilon_{\ket{W}}(d_I,d_O,d_E)
\end{align}
and the optimal pure-output strategy that attains this error.

\begin{lemma}\label{lem::error_pureoutput}
    A given pure-output approximate purification machine $\Qcal_{\ket{W}}:\Lcal (\Hcal_I\otimes\Hcal_O)\rightarrow \Lcal (\Hcal_A \otimes\Hcal_B \otimes \Hcal_E)$, such that
    \begin{equation}
        \Qcal_{\ket{W}}(C)=\ketbra{W}
    \end{equation}
    is a fixed isometry (independent of the input channel $C$), yields an average purification error given by
    \begin{equation}
        \epsilon_{\ket{W}}(d_I,d_O,d_E)=2d_I^2\left(1-\mathbb{E}_C\left\{F\left(\frac{C}{d_I},\frac{\tr_E\ketbra{W}}{d_I}\right)\right\}\right)=2\left(d_I^2-\mathbb{E}_C\left\{\tr(\sqrt{C}\sqrt{\tr_E\ketbra{W}})^2\right\}\right),
    \end{equation}
    where $F(\sigma,\rho)$ is the state fidelity.
\end{lemma}

\begin{proof}
For the fixed pure-output strategy, 
\begin{align}
    \epsilon_{\ket{W}}(d_I,d_O,d_E) & = \mathbb{E}_C \inf_{U_E} \left\{\tr[(\Qcal_{\ket{W}}(C))^2]+d_A^2-2\tr[ \Qcal_{\ket{W}}(C)(\id\otimes U_E)\ketbra{V(C)}(\id\otimes U_E^*)]\right\}\\
    & = \mathbb{E}_C\left\{\tr[\ketbra{W}^2]+d_A^2-2\sup_{U_E}\tr[\ketbra{W}(\id\otimes U_E)\ketbra{V(C)}(\id\otimes U_E^*)]\right\} \label{eq::pure_sup}\\
    & = \mathbb{E}_C\left\{2\left(d_A^2-d_A^2F\left(\frac{\tr_E\ketbra{V(C)}}{d_A},\frac{\tr_E\ketbra{W}}{d_A}\right)\right)\right\} \label{eq::post_uhlmann}\\
    & = 2d_I^2\left(1-\mathbb{E}_C\left\{F\left(\frac{C}{d_I},\frac{\tr_E\ketbra{W}}{d_I}\right)\right\}\right)\label{eq: fixed W out after uhlman} \\
    &= 2\left(d_I^2-\mathbb{E}_C\left\{\tr(\sqrt{C}\sqrt{\tr_E\ketbra{W}})^2\right\}\right),
\end{align}
where from Eq.~\eqref{eq::pure_sup} to~\eqref{eq::post_uhlmann} we used Uhlmann's theorem and recalling that $d_A=d_I$.
\end{proof}

We now prove tight lower and upper bounds for the error $\epsilon_{\ket{W}}(d_I,d_O,d_E)$ and show that the lower bound (best-case scenario) is attained by strategies that output an isometry $\ket{W}=\ket{\Omega}$ that is maximally entangled across the partition $AB|E$, and that the upper bound (worst-case scenario) is attained by strategies that output an isometry $\ket{W}=\ket{\Upsilon}\otimes\ket{\psi}$ that is separable across the partition $AB|E$.

\begin{theorem}[Minimum error of pure-output strategies -- Expanded]\label{thm::pure-output_app}
    The average purification error $\epsilon_{\ket{W}}(d_I,d_O,d_E)$ attained by any fixed pure-output strategy $\Qcal_{\ket{W}}(C)=\ketbra{W}$ satisfies
    \begin{equation}
      \epsilon_{\ket{\Omega}}(d_I,d_O,d_E) 
      \leq 
      \epsilon_{\ket{W}}(d_I,d_O,d_E) 
      \leq
      \epsilon_{\ket{\Upsilon}\otimes\ket{\psi}}(d_I,d_O,d_E),
    \end{equation}
    where 
    \begin{equation}
        \epsilon_{\ket{\Omega}}(d_I,d_O,d_E) =  2\left(
    d_I^2-\frac{1}{d_O}\,\EE_C\!\left\{(\tr\sqrt C)^2\right\}\right)
    \end{equation}
    is the error attained by a strategy the outputs an isometry $\ket{\Omega}$ that is maximally entangled across the $AB|E$ bipartition (i.e., a purification of the fully depolarizing channel), and is independent of which $\ket{\Omega}$, and 
    \begin{equation}
        \epsilon_{\ket{\Upsilon}\otimes\ket{\psi}}(d_I,d_O,d_E) = 2\left(d_I^2-\frac{d_I}{d_O}\right)
    \end{equation}
     is the error attained by a strategy the outputs an isometry $\ket{\Upsilon}\otimes\ket{\psi}$ that is separable across the $AB|E$ bipartition (i.e., a purification of a unitary channel), which is independent of $\ket{\Upsilon}$ and $\ket{\psi}$.

     Consequently, the minimum error of pure-output purification strategies over all possible fixed pure outputs $\ket{W}$ is given by
    \begin{align}
        \epsilon_\mathrm{pure}(d_I,d_O,d_E) &= \min_{\ket{W}} \epsilon_{\ket{W}}(d_I,d_O,d_E) = \epsilon_{\ket{\Omega}}(d_I,d_O,d_E) \\
        &=2\left(d_I^2-\frac1{d_O}\,\EE_C\!\left[(\tr\sqrt C)^2\right]\right)
    \end{align}
    and the optimal pure-output strategy is the one that maps all inputs to $\ketbra{\Omega}$, a purification of the fully depolarizing channel.
\end{theorem}
The proof of Thm.~\ref{thm::pure-output_app} is presented as Lemma~\ref{lem::upperb_pureoutput} for the upper bound and as Lemma~\ref{lem::lowerb_pureoutput} for the lower bound.

\begin{lemma}[Upper bound on the error of pure-output strategies]\label{lem::upperb_pureoutput}
   The average purification error $\epsilon_{\ket{W}}(d_I,d_O,d_E)$ attained by any pure-output strategy that outputs a fixed isometry is upper bounded by
   \begin{align}
       \epsilon_{\ket{W}}(d_I,d_O,d_E)\leq 2\left(d_I^2-\frac{d_I}{d_O}\right).
   \end{align}
   This bound is tight and attained by strategies that output isometries of the form $\ket{W}=\ket{\Upsilon}\otimes\ket{\psi}$, that is, isometries that are separable across the bipartition $AB|E$.
\end{lemma}

\begin{proof}
The upper bound can be derived via $F(X, Z)=\|\sqrt{X} \sqrt{Z}\|^2_1\geq\tr[X Z]=\|\sqrt{X} \sqrt{Z}\|_2^2 $ for full rank $X,Z\geq0$. Then, starting from Lemma~\ref{lem::error_pureoutput},
\begin{align}
    \epsilon_{\ket{W}}(d_I,d_O,d_E) &= 
    2\left(d_A^2-\mathbb{E}_C\{F(\tr_E\ketbra{W},C)\}\right) \\
    &\leq2\left(d_A^2-\tr[F_{AB,A'B'}(\tr_E[\ketbra{W}]\otimes  \mathbb{E}_{V}\{\tr_E[\ketbra{V}]\} )]\right) \label{eq:fidelity ineq}\\
    &=2\left(d_A^2-\tr[F_{AB,A'B'}(\tr_E[\ketbra{W}]\otimes  \frac{\id_{AB}}{d_B})]\right)\label{eq: average over fixed V}\\
    &=2\left(d_A^2-\tr[\frac{\tr_E\ketbra{W}}{d_B}]\right)\\
    &=2\left(d_A^2-\frac{d_A}{d_B}\right)\\
    &=2\left(d_I^2-\frac{d_I}{d_O}\right),
\end{align}
where $d_A=d_I$ and $d_B=d_O$, which is independent of the output isometric channel $\ketbra{W}$ and of $d_E$.

To show the attainability, we take $\ket{W}=\ket{\Upsilon}\otimes\ket{\psi}$ and compute
\begin{align}
    \epsilon_{\ket{\Upsilon}\otimes\ket{\psi}}(d_I,d_O,d_E) &= 2\left(d_I^2-\mathbb{E}_C\{F(\tr_E(\ketbra{\Upsilon}\otimes\ketbra{\psi}),C)\}\right) \\
    &= 2\left(d_I^2-\mathbb{E}_C\{F(\ketbra{\Upsilon},C)\}\right) \\
    &= 2\left(d_I^2-\tr(\ketbra{\Upsilon}\mathbb{E}_C\{C\})\right) \\
    &= 2\left(d_I^2-\tr(\ketbra{\Upsilon}\frac{\id}{d_O})\right) \\
    &=2\left(d_I^2-\frac{d_I}{d_O}\right).
\end{align}
Hence
\begin{equation}
     \epsilon_{\ket{W}}(d_I,d_O,d_E) \leq \epsilon_{\ket{\Upsilon}\otimes\ket{\psi}}(d_I,d_O,d_E) =2\left(d_I^2-\frac{d_I}{d_O}\right).
\end{equation}
\end{proof}

\begin{lemma}[Lower bound on the error of pure-output strategies]\label{lem::lowerb_pureoutput}
   The average purification error $\epsilon_{\ket{W}}(d_I,d_O,d_E)$ attained by any pure-output strategy that outputs a fixed isometry is lower bounded by
   \begin{align}
       \epsilon_{\ket{W}}(d_I,d_O,d_E)\geq 2\left(d_I^2-\frac1{d_O}\,\EE_C\!\left[(\tr\sqrt C)^2\right]\right).
   \end{align}
   This bound is tight and attained by strategies that output isometries of the form $\ket{W}=\ket{\Omega}$ where $\tr_E\ketbra{\Omega}=\id_{AB}/d_B$, that is, isometries that are maximally entangled across the bipartition $AB|E$.
\end{lemma}

\begin{proof}
Let
\begin{align}
    C_W\coloneqq\tr_E\ketbra{W}.
\end{align}
Now fix a unitary $U$ on $\Hcal_O$. Since $C$ is Haar-Stinespring random, its law is invariant under output conjugation:
\begin{align}
    C\overset{d}{=}(\id_I\otimes U)C(\id_I\otimes U^*).
\end{align}
Hence, using also the unitary invariance of the fidelity under common conjugation,
\begin{align}
    \EE_C \left\{F\left(\frac{C}{d_I},\frac{C_W}{d_I}\right)\right\}
    &= \EE_C \left\{F\left((\id_I\otimes U)\frac{C}{d_I}(\id_I\otimes U^*),(\id_I\otimes U)\frac{C_W}{d_I}(\id_I\otimes U^*)\right)\right\}\\
    &= \EE_C \left\{F\left(\frac{C}{d_I},(\id_I\otimes U)\frac{C_W}{d_I}(\id_I\otimes U^*)\right)\right\},
\end{align}
where in the second line we used that $(\id_I\otimes U)C(\id_I\otimes U^*)$ has the same distribution as $C$. Since this holds for every unitary $U$ on $\Hcal_O$, averaging over $U$ with respect to the Haar measure gives
\begin{align}
    \EE_C \left\{F\left(\frac{C}{d_I},\frac{C_W}{d_I}\right)\right\}
    &=
    \EE_C\int dU\,F\left(\frac{C}{d_I},(\id_I\otimes U)\frac{C_W}{d_I}(\id_I\otimes U^*)\right).
\end{align}

Now use concavity of the fidelity in its second argument for fixed \begin{align}
    \EE_C \left\{F\left(\frac{C}{d_I},\frac{C_W}{d_I}\right)\right\}
    &\le
    \EE_C\left\{F\!\left(\frac{C}{d_I},\int dU\,(\id_I\otimes U)\frac{C_W}{d_I}(\id_I\otimes U^*)\right)\right\}.
\end{align}
The Haar twirl on the output system is
\begin{align}
    \int dU\,(\id_I\otimes U)X(\id_I\otimes U^*)
    &=
    \frac{\tr_O X}{d_O}\otimes \id_O.
\end{align}
Applying this to $X=C_W$ and using $\tr_O C_W=\id_I$, we obtain
\begin{align}
    \int dU\,(\id_I\otimes U)C_W(\id_I\otimes U^*)
    &=
    \frac{\id_I\otimes \id_O}{d_O}
    =
    \frac{\id}{d_O}.
\end{align}
Therefore,
\begin{align}
    \EE_C \left\{F\left(\frac{C}{d_I},\frac{C_W}{d_I}\right
)\right\}
    &\le
    \EE_C\left\{F\left(\frac{C}{d_I},\frac{\id}{d_Id_O}\right) \right\}\\
    &=
    \frac{1}{d_Od_I^2}\,\EE_C\left\{\left(\tr\sqrt {C}\right)^2\right\}.
\end{align}
Substituting this into the definition of $\epsilon_{\ket{W}}(d_I,d_O,d_E)$ yields
\begin{align}
    \epsilon_{\ket{W}}(d_I,d_O,d_E)
    &=
    2\left(d_I^2-\EE_C \left\{F(C,C_W)\right\}\right)\\
    &\ge
    2\left(d_I^2-\frac1{d_O}\,\EE_C\left\{\left(\tr\sqrt C\right)^2\right\}\right),
\end{align}
which proves the lower bound.\\

To show the attainability, we take $\ket{W}=\ket{\Omega}$ and compute
\begin{align}
    \epsilon_{\ket{\Omega}}(d_I,d_O,d_E) &= 2\left(d_I^2-\mathbb{E}_C\left\{\tr[\sqrt{C\tr_E\ketbra{\Omega}}]^2\right\}\right) \\
    &= 2\left(d_I^2-\mathbb{E}_C\left\{\tr[\sqrt{C\,\frac{\id}{d_O}}]^2\right\}\right) \\
    &= 2\left(d_I^2-\frac1{d_O}\,\EE_C\left\{(\tr\sqrt C)^2\right\}\right).
\end{align}
\end{proof}

\section{Append-environment strategy}\label{app::appendenvironment}

\begin{figure}[H]
    \centering
    \includegraphics[width=0.4\linewidth]{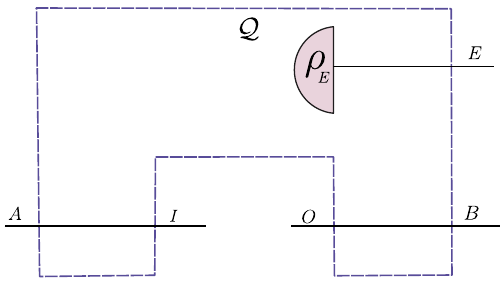}
    \caption{\textbf{Append-environment strategies.} Depiction of the family of strategies that forgets outputs the input channel and appends an ancillary state in the environment space.}
    \label{fig:append env}
\end{figure}

Here we analyze the average purification error attained by the single-copy append-environment strategies, which act as an identity superchannel on the input quantum channel $C$ and append a fixed state $\rho_E$ on the environment system. We recall that these strategies, defined in Sec.~\ref{sec::appendenv}, are described by a linear map $\Qcal_{\mathrm{app}-\rho_E}:\Lcal (\Hcal_I\otimes\Hcal_O)\rightarrow \Lcal (\Hcal_A \otimes\Hcal_B \otimes \Hcal_E)$ such that
\begin{equation}
        \Qcal_{\mathrm{app}-\rho_E}(C) = C \otimes \rho_E.
\end{equation}
Such strategies attain an error defined as
\begin{align}
    \epsilon_{\mathrm{app}-\rho_E}(d_I,d_O,d_E) &= \EE_C \min_{U_E} \norm{\Qcal_{\mathrm{app}-\rho_E}(C)-\ketbra{V_{U_E}(C)}}^2_2 \\
    & = \EE_C \min_{U_E}\norm{C\otimes\rho_E-\ketbra{V_{U_E}(C)}}^2_2.
\end{align}

We are interested in computing the error of the optimal strategy of this kind, that is, of 
\begin{equation}
     \epsilon_{\mathrm{app}}(d_I,d_O,d_E) = \min_{\rho_E}  \epsilon_{\mathrm{app}-\rho_E}(d_I,d_O,d_E).
\end{equation}

\begin{theorem}[Minimal error of append-environment strategies -- Expanded]\label{thm::append-env_app}
    The minimum error attained by an append-environment strategy, obtained by optimizing all such strategies $\Qcal_{\mathrm{app}-\rho_E}(C) = C \otimes \rho_E$ over the appended state $\rho_E$, is given by
    \begin{equation}
        \epsilon_{\mathrm{app}}(d_I,d_O,d_E) = d_I^2-\frac{d_O^2d_E^2-1}{d_Id_O(d_E^2-1)+d_I^2d_E(d_O^2-1)}\sum_{i=1}^{r} (\mathbb{E}_C\!\big[(c_i^\downarrow)^2\big])^2,
    \end{equation}
    where $\{c_i^\downarrow\}$ are the eigenvalues of $C$ in decreasing order, such that $C=\sum_{i=1}^r c_i^\downarrow\ketbra{i}$, and $r=\rank(C)$. 
    
    This quantity can furthermore be upper and lower bounded according to
    \begin{equation}
       d_I^2 - \frac{d_Id_O(d_E^2-1)+d_I^2d_E(d_O^2-1)}{d_O^2d_E^2-1} \leq \epsilon_\mathrm{app}(d_I,d_O,d_E) \leq d_I^2 - \frac{1}{d_E}\frac{d_Id_O(d_E^2-1)+d_I^2d_E(d_O^2-1)}{d_O^2d_E^2-1},
    \end{equation}
    where the upper bound is tight and attained by $\rho_E=\id/d_E$.
\end{theorem}

\begin{proof}
We begin by recalling the error functional related to a given append-environment strategy, given by
\begin{align}
    \epsilon_{\text{app}}(d_I,d_O,d_E) &= \min_{\rho_E}\EE_C\min_{U_E}\norm{C\otimes \rho_E-\ketbra{V_{U_E}(C)}}_2^2\\
    &=\min_{\rho_E}\mathbb{E}_{C}\left(d_A^2+\tr[C_{AB}^2]\tr[\rho_E^2]-2\max_{U_E}\tr[(C\otimes\rho_E)(\id\otimes U_E)\ketbra{V(C)}(\id\otimes U_E^*)]\right)\\
    &=\min_{\rho_E}\mathbb{E}_{C}\left(d_A^2+\tr[C_{AB}^2]\tr[\rho_E^2]-2\max_{U_E}\tr[(C\otimes U_E^*\rho_EU_E)\ketbra{V(C)}]\right).
    \end{align}
We now compute the third term by using the property $(A\otimes B)\ket{\Phi^+}=(AB^T\otimes \id)\ket{\Phi^+}$ and the von Neumann inequality $\max_{U}\tr[AUBU^*]=\sum_{i=1}^{r_max}\lambda_i(A)^\downarrow\lambda_i(B)^\downarrow$, for which
\begin{align}
    \max_{U_E}\tr[(C\otimes U_E^*\rho_EU_E)\ketbra{V(C)}]&=\max_{U_E}\tr[(C^2\otimes U_E^*\rho_EU_E)\ketbra{\Phi^+}]\\
    &=\max_{U_E}\tr[C^2(U_E^*\rho_EU_E)^T]\\
    &=\sum_{i=1}^{\rank(C)}(c_i^\downarrow)^2\lambda_i(\rho_E)^\downarrow.
\end{align}
Here we are using $C=\sum_{i=1}^r c_i^\downarrow$, with $c_1\geq c_2\geq...\geq c_r$ and $\rho_E=\sum_{i=1}^{d_E} \lambda_i(\rho_E)^\downarrow$, with $\lambda_1(\rho_E)\geq...\geq\lambda_r(\rho_E)$. For $C$ generated via Haar-Stinespring Choi vectors, their rank is $r=\min\{d_E,d_Ad_B\}$ as stated in Refs.~\cite{kukulski2021generating,Collins_Nechita_2016}.

The error then reads,
\begin{align}
    \epsilon_{\text{app}}(d_I,d_O,d_E) &= d_I^2+\min_{\rho_E}\{\mathbb E_C\tr[C^2]\tr[\rho_E^2]-2\mathbb E_C\sum_{i}^r (c_i^\downarrow)^2\lambda^{\downarrow}(\rho_E)\}.
\end{align}

Given that $\tr(\rho_E)=1$, we know $\sum_i \lambda_i^\downarrow(\rho_E)=1$. We define the ordered weights
\begin{align}
    w_i \;\coloneqq\;\mathbb{E}_C\!\big[(c_i^\downarrow)^2\big],
    \qquad i=1,\dots,r,
    \qquad
    w_{r+1}=\cdots=w_{d_E}\coloneqq0 \ \ (\text{if } d_E>r).
\end{align}

By the linearity of the expectation,
\begin{align}
    \mathbb{E}_C\!\left[\sum_{i=1}^{r} (c_i^\downarrow)^2\,\lambda_i^\downarrow(\rho_E)\right]
    \;=\;
    \sum_{i=1}^{r}\mathbb{E}_C\!\big[(c_i^\downarrow)^2\big]\,\lambda_i^\downarrow(\rho_E)
    \;=\;
    \sum_{i=1}^{d_E} w_i \lambda_i^\downarrow(\rho_E).
\end{align}
Thus, the minimization is equivalent to the convex program
\begin{align}\label{eq: optimization over eigenvalues}
    \epsilon_{\text{app}}(d_I,d_O,d_E) = \min_{\lambda_i^\downarrow(\rho_E)}
    \left\{
    d_I^2\;+\;\EE_C\tr[C^2]\sum_{i=1}^{d_E}\lambda_i^\downarrow(\rho_E)^2\;-\;2\sum_{i=1}^{d_E} w_i \lambda_i^\downarrow(\rho_E)
    \right\},
\end{align}
with $$\lambda_i^\downarrow(\rho_E)\in\mathbb{R}^{d_E}:\left\{\begin{array}{cc}
     \lambda_i^\downarrow(\rho_E)_1\ge\cdots\ge \lambda_i^\downarrow(\rho_E)_{d_E}\ge 0  \\
      \sum_i \lambda_i^\downarrow(\rho_E)=1
\end{array}\right.\,.$$

Consider first the simplex constraints $\lambda_i^\downarrow(\rho_E)\ge 0$ and $\sum_i \lambda_i^\downarrow(\rho_E)=1$ (the ordering constraint
will be checked a posteriori).
The Lagrangian for the task is
\begin{align}
    \Lcal(\lambda_i^\downarrow(\rho_E),\tau,\mu)
    =
    \EE_C\tr[C^2]\sum_{i=1}^{d_E}\lambda_i^\downarrow(\rho_E)^2 \;-\; 2\sum_{i=1}^{d_E} w_i \lambda_i^\downarrow(\rho_E)
    \;+\;\tau\Big(\sum_{i=1}^{d_E}\lambda_i^\downarrow(\rho_E)-1\Big)
    \;-\;\sum_{i=1}^{d_E}\mu_i \lambda_i^\downarrow(\rho_E),
\end{align}
with multipliers $\mu_i\ge 0$. Stationarity gives, for each $i$,
\begin{align}
    2\EE_C\tr[C^2]\lambda_i^\downarrow(\rho_E) - 2 w_i + \tau - \mu_i = 0.
\end{align}
Complementary slackness $\mu_i \lambda_i^\downarrow(\rho_E)=0$ implies the standard thresholding solution
\begin{align}
    \lambda_i^\downarrow(\rho_E)^{opt}
    =
    \frac{(w_i-\tau/2)_+}{\EE_C\tr[C^2]}
    \qquad (i=1,\dots,d_E),
\end{align}
where $(x)_+\coloneqq\max\{x,0\}$ and $\tau$ is chosen so that $\sum_i \lambda_i^\downarrow(\rho_E)^{opt}=1$.

For the random Haar-Stinespring ensemble under consideration, one has
\begin{align}
    \sum_{i=1}^{d_E} w_i
    =
    \mathbb{E}_C\!\left[\sum_{i=1}^{r} (c_i^\downarrow)^2\right]
    =
    \mathbb{E}_C\!\big[\mathrm{Tr}(C^2)\big].
\end{align}
Imposing $\sum_i \lambda_i^\downarrow(\rho_E)^{opt}=1$ therefore forces the threshold to be $0$, i.e.
$\tau/2=0$, since
\begin{align}
    \sum_{i=1}^{d_E}(w_i-\tau/2)_+
    \begin{cases}
        < \sum_i w_i=\EE_C\tr[C^2], & \text{if }\tau/2>0,\\
        > \sum_i w_i=\EE_C\tr[C^2], & \text{if }\tau/2<0.
    \end{cases}
\end{align}
Hence $\tau/2=0$ and
\begin{align}
    \lambda_i^\downarrow(\rho_E)^{opt}=\frac{w_i}{\EE_C\tr[C^2]}\qquad (i=1,\dots,d_E).
\end{align}
Since $w_1\ge\cdots\ge w_r\ge 0$ (ordered by definition) and $w_{r+1}=\cdots=w_{d_E}=0$,
the vector $\lambda(\rho_E)^{opt}$ is already non-increasing, so it is a feasible point.
Therefore any state $\rho_E^{opt}$ with spectrum
\begin{align}
    \lambda^\downarrow(\rho_E^{opt})
    =
    \left(\frac{w_1}{\EE_C\tr[C^2]},\dots,\frac{w_r}{\EE_C\tr[C^2]},0,\dots,0\right)
\end{align}
is a minimizer.

Substituting the optimal value in Eq.~\eqref{eq: optimization over eigenvalues},
\begin{align}\label{eq: error appendenv wo purity}
    \epsilon_{\text{app}}(d_I,d_O,d_E)
    &=
    d_I^2
    +\EE_C\tr[C^2]\sum_{i=1}^{d_E}\left(\frac{w_i}{\EE_C\tr[C^2]}\right)^2
    -2\sum_{i=1}^{d_E} w_i\left(\frac{w_i}{\EE_C\tr[C^2]}\right)
    =
    d_I^2-\frac{1}{\EE_C\tr[C^2]}\sum_{i=1}^{r} w_i^2
\end{align}
with $w_i=\mathbb{E}_C\!\big[(c_i^\downarrow)^2\big]$.

Using Eq.~\eqref{eq::avgpurity_app} for $\EE_C\tr[C^2]$, we conclude that
\begin{align}\label{eq:: error appendenv}
    \epsilon_{\text{app}}(d_I,d_O,d_E) = d_I^2-\frac{d_O^2d_E^2-1}{d_Id_O(d_E^2-1)+d_I^2d_E(d_O^2-1)}\sum_{i=1}^{r} \left(\EE_C[(c_i^\downarrow)^2]\right)^2.
\end{align}

We can derive a lower bound for this quantity via $\sum_{i=1}^{r} w_i^2\leq (\sum_{i=1}^{r} w_i)^2=\EE_C\{\tr[C^2]\}^2$, therefore
\begin{equation}
    \epsilon_{\text{app}}(d_I,d_O,d_E) \geq d_I^2-\frac{d_Id_O(d_E^2-1)+d_I^2d_E(d_O^2-1)}{d_O^2d_E^2-1}.
\end{equation}

An upper bound can be derived by taking $\rho_E=\id/d_E$, in which case
\begin{align}
    \epsilon_{\text{app-}\id/d_E}(d_I,d_O,d_E) &= \EE_C \min_{U_E}\norm{C\otimes\frac{\id}{d_E}-\ketbra{V_{U_E}(C)}}^2_2 \\
    &= d_I^2 + \frac{1}{d_E}\EE_C\{\tr(C^2)\} - 2 \EE_C \max_{U_E}\tr[\left(C\otimes\frac{U_E^*\id U_E}{d_E}\right)\ketbra{V(C)}]\\
    &=  d_I^2 + \frac{1}{d_E}\EE_C\{\tr(C^2)\} - \frac{2}{d_E} \EE_C\tr[C\tr_E\ketbra{V(C)}]\\
    &=  d_I^2 + \frac{1}{d_E}\EE_C\{\tr(C^2)\} - \frac{2}{d_E} \EE_C\{\tr(C^2)\}\\
    &=  d_I^2 - \frac{1}{d_E}\EE_C\{\tr(C^2)\} \\
    &=  d_I^2 - \frac{1}{d_E}\frac{d_Id_O(d_E^2-1)+d_I^2d_E(d_O^2-1)}{d_O^2d_E^2-1} \label{eq::append_maxmixedstate}
\end{align}
Hence,
\begin{equation}
    \epsilon_{\text{app}}(d_I,d_O,d_E) \leq d_I^2 - \frac{1}{d_E}\frac{d_Id_O(d_E^2-1)+d_I^2d_E(d_O^2-1)}{d_O^2d_E^2-1}.
\end{equation}

\end{proof}

The particular case where the appended environment state $\rho_E=\ketbra{\psi}$ is pure can be of particular interest. In this case, for strategies $\Qcal_{\mathrm{app}-\ket{\psi}}(C)=C\otimes\ketbra{\psi}$, we would like to compute the minimum error over all pure states $\ket{\psi}$, given by
\begin{align}
    \epsilon_{\text{app-pure}}(d_I,d_O,d_E) \coloneqq \min_{\ket{\psi}} \EE_C \min_{U_E}\norm{C\otimes\ketbra{\psi}-\ketbra{V_{U_E}(C)}}^2_2
\end{align}

\begin{proposition}[Minimum error of append-pure state strategy]\label{prop::app_purestate}
    The minimum error attained by a strategy that appends a pure state onto the environment system is 
    \begin{equation}
        \epsilon_{\mathrm{app-pure}}(d_I,d_O,d_E) = d_I^2 + \frac{d_Id_O(d_E^2-1)+d_I^2d_E(d_O^2-1)}{d_O^2d_E^2-1} - 2\,\EE_C\{c_\mathrm{max}^2\},
    \end{equation}
    where $c_\mathrm{max}$ is the maximal eigenvalue of $C$.
\end{proposition}
\begin{proof}
We start from expression Eq.~\eqref{eq: optimization over eigenvalues} noting that in this case $\lambda_1^\downarrow(\rho_E)=1$ and the rest are zero,
\begin{align}\label{eq: pure ancilla append error}
    \epsilon_{\text{app}}(d_I,d_O,d_E) &= 
    d_I^2\;+\;\EE_C\tr[C^2]\;-\;2w_1\\
    &=d_I^2\;+\frac{d_Id_O(d_E^2-1)+d_I^2d_E(d_O^2-1)}{d_O^2d_E^2-1}-\;2\EE_C[c_{\max}^2]
\end{align}

\end{proof}
The quantity $\EE_C\{c_\mathrm{max}^2\}$ respects
\begin{equation}
    \frac{1}{d_O^2}\leq \EE_C\{c_\mathrm{max}^2\} \leq d_I^2,
\end{equation}
implying
\begin{equation}
    \epsilon_{\text{app-pure}}(d_I,d_O,d_E) \leq  \frac{d_Id_O(d_E^2-1)+d_I^2d_E(d_O^2-1)}{d_O^2d_E^2-1} + d_I^2  - \frac{2}{d_O^2}.
\end{equation}
From this we conclude that the error attained by the optimal append-environment strategy that prepares a pure state for the environment is incomparable with that which prepares the maximally mixed state for the environment, given by Eq.~\eqref{eq::append_maxmixedstate}.

\section{Map-to-depolarizing strategy}\label{app::maptodepol}

\begin{figure}[H]
    \centering
    \includegraphics[width=0.4\linewidth]{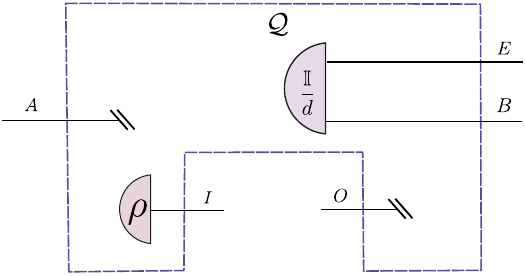}
    \caption{\textbf{Map-to-depolarizing strategy.} Depiction of a strategy mapping every input map into the fully depolarizing channel.}
    \label{fig: map to depolarizing}
\end{figure}

Here we analyze the average purification error attained by the strategy that maps all input quantum channels to the fully depolarizing channel. We recall that these strategies, defined in Sec.~\ref{sec::maptodep}, are described by a linear map $\Qcal_\mathrm{dep}:\Lcal(\Hcal_I\otimes\Hcal_O)\rightarrow \Lcal (\Hcal_A \otimes\Hcal_B \otimes \Hcal_E)$ such that
\begin{equation}
        \Qcal_\mathrm{dep}(C) = \tr(C)\frac{\id}{d_Id_Od_E} = \frac{\id}{d_Od_E}.
\end{equation}
The fully depolarizing channel $\widetilde{D}(\rho)=\tr(\rho)\id/d_O$ is the least–extremal (“most mixed”) channel and the barycenter of the space of CPTP maps. By symmetry, for any unitarily invariant distance, $\mathcal{D}$ is isotropic (on average, equally distant from all directions) making this an interesting strategy for single-copy purification. Such strategy attains an error defined as
\begin{align}
    \epsilon_\mathrm{dep}(d_I,d_O,d_E) &= \EE_C \min_{U_E} \norm{\Qcal_\mathrm{dep}(C)-\ketbra{V_{U_E}(C)}}^2_2 \\
    & = \EE_C \min_{U_E}\norm{\frac{\id}{d_Od_E}-\ketbra{V_{U_E}(C)}}^2_2.
\end{align}

\begin{theorem}[Error of the map-to-depolarizing strategy -- Expanded]\label{thm::maptodepol_app}
    The average purification error attained by a strategy $\Qcal_\mathrm{dep}(C)=\id/d_Od_E$ which maps all input channels to the fully depolarizing channel is
    \begin{equation}
        \epsilon_\mathrm{dep}(d_I,d_O,d_E)=d_I^2-\frac{d_I}{d_Od_E}.
    \end{equation}
\end{theorem}

\begin{proof}
\begin{align}
    \epsilon_\mathrm{dep}(d_I,d_O,d_E) &= d_I^2 + \tr[\frac{\id}{d_O^2d_E^2}] - \frac{2}{d_Od_E}\EE_C\max_{U_E}\tr[(\id\otimes U_E)\ketbra{V(C)}(\id\otimes U_E^*)] \\
    &= d_I^2 + \frac{d_I}{d_Od_E} - 2\frac{d_I}{d_Od_E} \\
    &= d_I^2 - \frac{d_I}{d_Od_E}. 
\end{align}
where the negative term in the third equality is obtained by cyclicity of the trace and normalization of the Choi operators.
\end{proof}

It can then be seen that this strategy, although it seemed reasonable, it is very close to the maximal possible error.

\section{General error upper bound}\label{app::genuperbound}
Here we derive the general error upper bound in Sec.~\ref{sec::genupperbound}, obtained by performing an average instead of a minimization over the set of unitaries on the environment. That is, we compute
\begin{align}
    \epsilon_{\mathrm{avg-}U_E}(d_I,d_O,d_E) &= \min_\Qcal \EE_C \EE_{U_E} \norm{\Qcal(C^{\otimes k}) - \ketbra{V_{U_E}(C)}}^2_2 \label{eq::epsilonavg_app}\\
    &\geq \min_\Qcal \EE_C \min_{U_E} \norm{\Qcal(C^{\otimes k}) - \ketbra{V_{U_E}(C)}}^2_2 \\
    &= \epsilon(d_I,d_O,d_E;k). \\
\end{align}
As we will see, the solution for $\epsilon_{\mathrm{avg-}U_E}(d_I,d_O,d_E)$ is independent of $k$.

\begin{theorem}[General error upper bound -- Expanded]\label{thm::genupperbound_app}
    The upper bound $\epsilon_{\mathrm{avg-}U_E}(d_I,d_O,d_E)$ defined in Eq.~\eqref{eq::epsilonavg_app} of the average purification error $\epsilon(d_I,d_O,d_E;k)$ is exactly
    \begin{equation}
        \epsilon_{\mathrm{avg-}U_E}(d_I,d_O,d_E) = d_I^2-\frac{1}{d_E}\frac{d_I^2d_E(d_O^2-1)+d_Id_O(d_E^2-1)}{d_E^2d_O^2-1}
    \end{equation}
    for all $k$.
\end{theorem}

\begin{proof}
\begin{align}
    \epsilon_{\mathrm{avg-}U_E}(d_I,d_O,d_E) &= \min_\Qcal \EE_C \EE_{U_E} \norm{\Qcal(C^{\otimes k}) - \ketbra{V_{U_E}(C)}}^2_2 \\
    &= d_I^2 + \min_\Qcal \EE_C \left\{ \tr[Q(C^{\otimes k})^2] - 2\EE_{U_E} \tr[Q(C^{\otimes k})(\id\otimes U_E)\ketbra{V(C)}(\id\otimes U_E^*)]\right\} \\
    &= d_I^2 + \min_\Qcal \EE_C \left\{ \tr[Q(C^{\otimes k})^2] - 2\tr[Q(C^{\otimes k})\left(\tr_E \ketbra{V(C)} \otimes \frac{\id}{d_E}\right)]\right\} \\
    &= d_I^2 + \min_\Qcal \EE_C \left\{ \tr[Q(C^{\otimes k})^2] - 2\tr[Q(C^{\otimes k})\left(C \otimes \frac{\id}{d_E}\right)]\right\}.
\end{align}

By noticing that since
\begin{align}
    \norm{\Qcal(C^{\otimes k}) - C\otimes\frac{\id}{d_E}}_2^2 &= \tr[\Qcal(C^{\otimes k})^2] + \tr[\left(C\otimes\frac{\id}{d_E}\right)^2] - 2\tr[\Qcal(C^{\otimes k})\left(C\otimes\frac{\id}{d_E}\right)]
\end{align}
one has that 
\begin{align}
  \tr[\Qcal(C^{\otimes k})^2] - 2\tr[\Qcal(C^{\otimes k})\left(C\otimes\frac{\id}{d_E}\right)] &= \norm{\Qcal(C^{\otimes k}) - C\otimes\frac{\id}{d_E}}_2^2 - \frac{1}{d_E}\tr(C^2), 
\end{align}
and, consequently,
\begin{align}
    \epsilon_{\mathrm{avg-}U_E}(d_I,d_O,d_E) &=  d_I^2 + \min_\Qcal \EE_C \left\{ \norm{\Qcal(C^{\otimes k}) - C\otimes\frac{\id}{d_E}}_2^2 - \frac{1}{d_E}\tr(C^2) \right\}
\end{align}
Taking $\Qcal(C^{\otimes k}) = C\otimes\frac{\id}{d_E}$ being the strategy that outputs the input quantum channel unchanged and appends a maximally mixed state onto the environment, the first term in the minimization goes to zero, hence
\begin{align}
    \epsilon_{\mathrm{avg-}U_E}(d_I,d_O,d_E) &=  d_I^2 - \frac{1}{d_E} \EE_C \left\{\tr(C^2) \right\} \\
    &= d_I^2-\frac{1}{d_E}\frac{d_I^2d_E(d_O^2-1)+d_Id_O(d_E^2-1)}{d_E^2d_O^2-1}.
\end{align}

\end{proof}

It can be easily seen the optimal strategy in this case has the same form as the so-called random dilation superchannel defined in Refs.~\cite{girardi2025random2,yoshida2025}.

\section{Estimation-based many-copy strategy}\label{app::estimation}

Let us perform an estimation-based protocol as the one described in Refs.~\cite{girardi2025random} and~\cite{PelecanosEtAl2025}, where they implement the following steps:

\begin{enumerate}
   
    \item
    Run the mixed-state tomography algorithm of
    Ref.~\cite[(Theorem~1.3)]{PelecanosEtAl2025} on $k$ copies of
    $C/d_I$.  The output of one run is a state estimate $ \frac{\widehat C_y}{d_I},$
    where $y\in\mathcal Y$ denotes the classical outcome of the tomography
    procedure.

    \item
    For each outcome $y$, prepare a fixed-gauge purification of
    $\widehat C_y/d_I$, choosing for instance the fixed local environment
    unitary $U_E=\id_E$.  Thus the final output in one run is$\ketbra{\widehat V_y(C)},$
    with $\tr_E \ketbra{\widehat V_y(C)}=\frac{\widehat C_y}{d_I}.$
\end{enumerate}

The preceding description is the output of a single run, conditioned on the
classical tomography outcome $y$.  If the learning machine is regarded as a
deterministic quantum operation and the classical register $y$ is forgotten,
then its actual deterministic output is not one of the pure states
$\ketbra{\widehat V_y(C)}$; rather, it is the averaged state
\begin{align}
    \Qcal_\mathrm{tomo}(C^{\otimes k})
    :=
    \int_{\mathcal Y}
    dP(y\mid C)\,
    \ketbra{\widehat V_y(C)} ,
\end{align}
where $P(\cdot\mid C)$ is the outcome distribution of the tomography
algorithm on input $(C/d_I)^{\otimes k}$.  In particular,
$\Qcal_\mathrm{tomo}(C^{\otimes k})$ is generally mixed, and
\begin{align}
    \tr_E \left[\Qcal_\mathrm{tomo}(C^{\otimes k})\right]
    =
    \int_{\mathcal Y}
    dP(y\mid C)\,
    \frac{\widehat C_y}{d_I}.
\end{align}
\begin{theorem}[Estimation-based strategy]\label{thm::estimation_app}
There exist an estimation-based deterministic purification machine $\Qcal_\mathrm{tomo}$ that takes in $k$ copies of an unknown $C$ , performs the strategy explained above and outputs $\int_{\mathcal Y}
dP(y\mid C)\,
\ketbra{\widehat V_y(C)}$. 
Then the average error for this strategy is upper bounded by
\begin{align}
     \epsilon_{\mathrm{tomo}}(d_I,d_O,d_E;k)
    \leq
    \frac{4\kappa d_I^2
    \left(
        d_I d_O\,\min\{d_E,d_Id_O\}
        +
        \log(1/\delta)
    \right)}{k}
    +
    O\!\left(
        \frac{d_I^2}{k^2}
        \EE_C\left[
            \left(
                \operatorname{rank}(C)d_I d_O+\log(1/\delta)
            \right)^2
        \right]
    \right)
    +
    O(d_I^2\delta).
\end{align}
where $\kappa>0$ is a universal constant. Here $1-\delta$ denotes the success probability of the estimation strategy of Ref.~\cite{PelecanosEtAl2025}: with probability at least $1-\delta$, the output is an estimate $\hat C$ satisfying
\begin{align}\label{eq:epsilon-k}
    F(C/d_I,\hat C/d_I) \geq 1-\varepsilon_k,
\end{align}
with
\begin{align}
    \varepsilon_k
    :=
    \min\left\{1,\kappa\,
    \frac{
        \operatorname{rank}(C)d_I d_O+\log(1/\delta)
    }{k}\right\}.
\end{align}
\end{theorem}

\begin{proof}

Let us start by writing the error functional for this strategy based on Eq. \eqref{eq::avg_error_def},
\begin{align}\label{eq:: def error tomo}
    &\epsilon_\mathrm{tomo}(d_I,d_O,d_E;k) \nonumber\\
    &=\EE_C\min_{U_E}\norm{\Qcal_\mathrm{tomo}(C^{\otimes k})-(\id\otimes U_E)\ketbra{V(C)}(\id\otimes U_E^*)}_2^2\\
    &=\EE_C\left\{\tr\left[\left(\int_{\mathcal Y}
    dP(y\mid C)\,
    \ketbra{\widehat V_y(C)}\right)^2 \right]+d_I^2-2\max_{U_E}\tr\left[\int_{\mathcal Y}
    dP(y\mid C)\,
    \ketbra{\widehat V_y(C)} \ketbra{V_{U_E}(C)}\right]\right\}
\end{align}

We next focus on the third term and lower bound it by fixing the unitary optimization, then
\begin{align}\label{eq:: upperbound fix gauge}
    \max_{U_E}&\tr\left[\int_{\mathcal Y}
    dP(y\mid C)\,
    \ketbra{\widehat V_y(C)} (\id\otimes U_E) \ketbra{V(C)}(\id\otimes U_E^*)\right]\\
    \geq&\, \tr\left[\int_{\mathcal Y}
    dP(y\mid C)\,
    \ketbra{\widehat V_y(C)} \ketbra{V(C)}\right]
\end{align}
Theorem~1.3 of Ref.~\cite{PelecanosEtAl2025} states that, for a rank-\(r\) state on a \(d\)-dimensional Hilbert space, \(k\) copies
are sufficient whenever
\begin{align}
    k
    =
    O\!\left(
        \frac{rd+\log(1/\delta)}{\varepsilon}
    \right)
\end{align}
to obtain an estimate with fidelity at least \(1-\varepsilon\) with probability
at least \(1-\delta\).  In our setting,
\begin{align}
    r=\operatorname{rank}(C),
    \qquad
    d=d_I d_O .
\end{align}
Equivalently, there exists a universal constant \(\kappa>0\) such that the
guarantee holds whenever
\begin{align}
    k
    \ge
    \kappa\,
    \frac{
        \operatorname{rank}(C)d_I d_O+\log(1/\delta)
    }{\varepsilon}.
\end{align}
For a fixed number \(k\) of copies, we therefore define
\begin{align}\label{eq:: }
    \varepsilon_k
    :=
    \min\left\{1,\kappa\,
    \frac{
        \operatorname{rank}(C)d_I d_O+\log(1/\delta)
    }{k}\right\},
\end{align}
which then it implies that
\begin{align}
    F\!\left(
        \frac{C}{d_I},
        \frac{\widehat C_y}{d_I}
    \right)
    \ge 1-\varepsilon_k
\end{align}
with probability at least \(1-\delta\) over \(y\sim P(\cdot\mid C)\).
Thus, we define
\begin{align}
    \mathcal S_C
    :=
    \left\{
        y\in\mathcal Y:
        F\!\left(
            \frac{C}{d_I},
            \frac{\widehat C_y}{d_I}
        \right)
        \ge 1-\varepsilon_k
    \right\}
\end{align}
as the set of outcomes from the tomography that are close to the input.\\
Then
\begin{align}
    P(\mathcal S_C\mid C)\ge 1-\delta,
    \qquad
    P(\mathcal Y\setminus \mathcal S_C\mid C)\le \delta.
\end{align}
Since the purifications are
chosen with a fixed coherent gauge, this implies
\begin{align}
    \left|
    \braket{\widehat V_y(C)}{V(C)}
    \right|^2
    \ge
    d_I^2(1-\varepsilon_k)^2,
    \qquad
    y\in\mathcal S_C,
\end{align}
where the factor \(d_I^2\) comes from the unnormalized Choi-vector
normalization.
For the remaining outcomes \(y\notin\mathcal S_C\), the theorem gives no
closeness guarantee.  However, their contribution is still nonnegative:
\begin{align}
    \left|
        \braket{\widehat V_y(C)}{V_{U_E}(C)}
    \right|^2
    \ge 0.
\end{align}
Therefore,
\begin{align}
    &\int_{\mathcal Y}
    dP(y\mid C)\,
    \tr\left[\ketbra{\widehat V_y(C)}\ketbra{V(C)}\right]
    \\
    &=\int_{\mathcal S_C}
    dP(y\mid C)\,
    \tr\left[\ketbra{\widehat V_y(C)}\ketbra{V(C)}\right]
    +
    \int_{\mathcal Y\setminus\mathcal S_C}
    dP(y\mid C)\,
    \tr\left[\ketbra{\widehat V_y(C)}\ketbra{V(C)}\right]\\
    &=
    \int_{\mathcal S_C}
    dP(y\mid C)\,
    \left|
    \braket{\widehat V_y(C)}{V_{U_E}(C)}
    \right|^2
    +
    \int_{\mathcal Y\setminus\mathcal S_C}
    dP(y\mid C)\,
    \left|
    \braket{\widehat V_y(C)}{V_{U_E}(C)}
    \right|^2
    \\
    &\ge
    \int_{\mathcal S_C}
    dP(y\mid C)\,
    d_I^2(1-\varepsilon_k)^2
    \\
    &=
    d_I^2(1-\varepsilon_k)^2
    P(\mathcal S_C\mid C)
    \\
    &\ge
    d_I^2(1-\varepsilon_k)^2(1-\delta) .
\end{align}

We plug this lower bound back to Eq.~\eqref{eq:: def error tomo},

\begin{align}
    \epsilon_\mathrm{tomo}(d_I,d_O,d_E;k)\leq \EE_C \left\{\tr\left[\left(\int_{\mathcal Y}
    dP(y\mid C)\,
    \ketbra{\widehat V_y(C)}\right)^2 \right]+d_I^2-2d_I^2(1-\delta)(1-\varepsilon_k)^2\right\}.
\end{align}
We can finally upper-bound the first term by the maximum purity possible,
\begin{align}
    \epsilon_\mathrm{tomo}(d_I,d_O,d_E;k)&\leq 2d_I^2-2d_I^2(1-\delta)\EE_C\{(1-\varepsilon_k)^2\}\\
    &=2d_I^2-2d_I^2(1-\delta)(1-2\EE_C\{\varepsilon_k\}+\EE_C\{\varepsilon_k^2\})\\
    &=4d_I^2\EE_C\{\varepsilon_k\}-2d_I^2\EE_C\{\varepsilon_k^2\}+2d_I^2\delta(1-2\EE_C\{\varepsilon_k\}+\EE_C\{\varepsilon_k^2\})\\
    &=4d_I^2\EE_C\left\{\min\left\{1,\kappa\,
    \frac{
    \operatorname{rank}(C)d_I d_O+\log(1/\delta)
    }{k}\right\}\right\}-2d_I^2\EE_C\{\varepsilon_k^2\}+2d_I^2\delta(1-2\EE_C\{\varepsilon_k\}+\EE_C\{\varepsilon_k^2\})\\
    &=4d_I^2\min\left\{1,\kappa\,
    \frac{
    \min\{d_E,d_Id_O\}d_I d_O+\log(1/\delta)
    }{k}\right\}-2d_I^2\EE_C\{\varepsilon_k^2\}+2d_I^2\delta(1-2\EE_C\{\varepsilon_k\}+\EE_C\{\varepsilon_k^2\})
\end{align}

The last two contributions are negligible at the \(1/k\) level. Indeed,
\begin{align}
    \EE_C\{\varepsilon_k^2\}
    =
    O\!\left(
        \frac{1}{k^2}
        \EE_C\left\{
            \left(
                \operatorname{rank}(C)d_I d_O+\log(1/\delta)
            \right)^2
        \right\}
    \right),
\end{align}
while
\begin{align}
    2d_I^2\delta
    \left(
        1-2\EE_C\{\varepsilon_k\}+\EE_C\{\varepsilon_k^2\}
    \right)
    =
    O(d_I^2\delta).
\end{align}
Therefore,
\begin{align}
\begin{split}
    &\epsilon_{\mathrm{tomo}}(d_I,d_O,d_E;k) \\
    &\leq
    4 d_I^2\min\left\{1,\kappa\,
    \frac{
        \min\{d_E,d_Id_O\}d_I d_O+\log(1/\delta)
    }{k}\right\}
    +
    O\!\left(
        \frac{d_I^2}{k^2}
        \EE_C\left\{
            \left(
                \operatorname{rank}(C)d_I d_O+\log(1/\delta)
            \right)^2
        \right\}
    \right)
    +
    O(d_I^2\delta).
\end{split}
\end{align}

\end{proof}

\end{document}